\documentclass{article}

\usepackage{arxiv}
\usepackage[utf8]{inputenc}

\usepackage{hyperref}

\usepackage{longtable}
\usepackage{authblk}
\usepackage[caption = false]{subfig}
\usepackage{graphicx}
\usepackage{amssymb}
\usepackage{amsmath}
\usepackage{amsthm}
\usepackage[noend]{algpseudocode}
\usepackage{algorithm}
\usepackage{titletoc}
\usepackage{titlesec}
\usepackage[utf8]{inputenc}
\usepackage{enumitem}
\setitemize{noitemsep,topsep=0pt,parsep=0pt,partopsep=0pt,leftmargin=0.8cm}
\usepackage{hhline}

\makeatletter
\renewcommand{\Function}[2]{%
	\csname ALG@cmd@\ALG@L @Function\endcsname{#1}{#2}%
	\def\jayden@currentfunction{#1}%
}
\newcommand{\funclabel}[1]{%
	\@bsphack
	\protected@write\@auxout{}{%
		\string\newlabel{#1}{{\jayden@currentfunction}{\thepage}}%
	}%
	\@esphack
}
\makeatother

\newcommand{\onlay}{{\rm ONLAY}}
\newcommand{\stakedag}{{\rm StakeDag}}
\newcommand{\stair}{{\rm STAIR}}

\newtheorem{thm}{Theorem}[section]
\newtheorem{lem}[thm]{Lemma}
\newtheorem{prop}[thm]{Proposition}
\newtheorem{defn}{Definition}[section]
\usepackage[english]{babel}

\newcommand{\dfnn}[2]{ \textbf{\emph{[#1]}} {#2}}

\renewcommand{\vec}[1]{\mathbf{#1}}
\newcommand{\eself}{\hookrightarrow^{s}}
\newcommand{\eref}{\hookrightarrow^{r}}
\newcommand{\eancestor}{\hookrightarrow^{a}} 
\newcommand{\eselfancestor}{\hookrightarrow^{sa}} 
\newcommand{\erefz}{\hookrightarrow}
\newcommand{\efork}{\pitchfork}

\newcommand{\dom}{\gg}
\newcommand{\sdom}{\gg^{s}}
\newcommand{\domf}{\gg^{f}}

\newcommand{\hibefore}{\mapsto}
\newcommand{\hbefore}{\rightarrow}
\newcommand{\concur}{\parallel}

\makeatletter
\def\BState{\State\hskip-\ALG@thistlm}
\makeatother

\title{Lachesis: Scalable Asynchronous BFT on DAG Streams}

\author{Quan Nguyen, Andre Cronje, Michael Kong, Egor Lysenko, Alex Guzev}
\affil{FANTOM}

\begin{document}
\maketitle

\begin{abstract}
This paper consolidates the core technologies and key concepts of our novel Lachesis consensus protocol and Fantom Opera platform, which is permissionless, leaderless and EVM compatible.

We introduce our new protocol, so-called Lachesis, for distributed networks achieving Byzantine fault tolerance (BFT)~\cite{lachesis01}. Each node in Lachesis protocol operates on a local block DAG, namely \emph{OPERA DAG}. 
Aiming for a low time to finality (TTF) for transactions, our general model considers DAG streams of high speed but asynchronous events. 
We integrate Proof-of-Stake (PoS) into a DAG model in Lachesis protocol to improve performance and security. Our general model of trustless system leverages participants' stake as their validating power~\cite{stakedag}.
Lachesis's consensus algorithm uses Lamport timestamps, graph layering and concurrent common knowledge to guarantee a consistent total ordering of event blocks and transactions. In addition, Lachesis protocol allows dynamic participation of new nodes into Opera network. Lachesis optimizes DAG storage and processing time by splitting local history into checkpoints (so-called epochs). We also propose a model to improve stake decentralization, and network safety and liveness ~\cite{stairdag}.

Built on our novel Lachesis protocol, Fantom's Opera platform is a public, leaderless, asynchronous BFT, layer-1 blockchain, with guaranteed deterministic finality. Hence, Lachesis protocol is suitable for distributed ledgers by leveraging asynchronous partially ordered sets with logical time ordering instead of blockchains.
We also present our proofs into a model that can be applied to abstract asynchronous distributed system.
\end{abstract}


\newpage
\pagenumbering{arabic} 
\tableofcontents 
\newpage

\section{Introduction}\label{ch:intro}
Interests in blockchains and distributed ledger technologies have surged significantly since Bitcoin~\cite{bitcoin08} was introduced in late 2008. Bitcoin and blockchain technologies have enabled numerous opportunities for business and innovation.
Beyond the success over cryptocurrency, the decentralization principle of blockchains has attracted numerous applications across different domains from financial, healthcare to logistics. Blockchains provide immutability and transparency of blocks that facilitate highly trustworthy, append-only, transparent public distributed ledgers and area promising solution for building trustless systems. Public blockchain systems support third-party auditing and some of them offer a high level of anonymity. 

Blockchain is generally a distributed database system, which stores a database of all transactions within the network and replicate it to all participating nodes. A distributed consensus protocol running on each node to guarantee consistency of these replicas so as to maintain a common transaction ledger.
In a distributed database system, \emph{Byzantine} fault tolerance (BFT)~\cite{Lamport82} is paramount to reliability of the system to make sure it is tolerant up to one-third of the participants in failure. 

Consensus algorithms ensure the integrity of transactions over the distributed network, and they are equivalent to the proof of BFT~\cite{randomized03, paxos01}. 
In practical BFT (pBFT) can reach a consensus for a block once the block is shared with other participants and the share information is further shared with others \cite{zyzzyva07, honey16}.

Despite of the great success, blockchain systems are still facing some limitations. 
Recent advances in consensus algorithms \cite{ppcoin12, dpos14, dagcoin15, algorand17, sompolinsky2016spectre, PHANTOM08} have improved the consensus confirmation time and power consumption over blockchain-powered distributed ledgers. 

\emph{Proof of Work} (PoW): is introduced in the Bitcoin's Nakamoto consensus protocol~\cite{bitcoin08}. Under PoW, validators are randomly selected based on their computation power. PoW protocol requires exhaustive computational work and high demand of electricity from participants for block generation. It also requires longer time for transaction confirmation.

\emph{Proof Of Stake} (PoS): leverages participants' stakes for selecting the creator of the next block~\cite{ppcoin12,dpos14}. Validators have voting power proportional to their stake. PoS requires less power consumption and is more secure than PoW, because the stake of a participant is voided and burnt if found dishonest.

\emph{DAG} (Directed Acyclic Graph): DAG based currency was first introduced in DagCoin paper~\cite{dagcoin15}. DAG technology allows cryptocurrencies to function similarly to those that utilize blockchain technology without the need for blocks and miners. 
DAG-based approaches utilize directed acyclic graphs (DAG)~\cite{dagcoin15, sompolinsky2016spectre, PHANTOM08, PARSEC18, conflux18} to facilitate consensus. 
Examples of DAG-based consensus algorithms include Tangle~\cite{tangle17}, Byteball~\cite{byteball16}, and Hashgraph~\cite{hashgraph16}. Tangle selects the blocks to connect in the network utilizing accumulated weight of nonce and Monte Carlo Markov Chain (MCMC). Byteball generates a main chain from the DAG and reaches consensus through index information of the chain. Hashgraph connects each block from a node to another random node. Hashgraph searches whether 2/3 members can reach each block and provides a proof of Byzantine fault tolerance via graph search.

\subsection{Motivation}
Practical Byzantine Fault Tolerance (pBFT) allows all nodes to successfully reach an agreement for a block (information) when a Byzantine node exists \cite{Castro99}. In pBFT, consensus is reached once a created block is shared with other participants and the share information is shared with others again \cite{zyzzyva07, honey16}. After consensus is achieved, the block is added to the participants’ chains~\cite{Castro99, Blockmania18}.
It takes $O(n^4)$ for pBFT, where $n$ is the number of participants.

DAG-based approaches~\cite{dagcoin15, sompolinsky2016spectre, PHANTOM08, PARSEC18, conflux18} and 
Hashgraph~\cite{hashgraph16} leverage block DAG of event blocks to reach consensus. The propose to use virtual voting on the local block DAG to determine consensus, yet their approach has some limitations.
First, the algorithm operates on a known network comprised of known authoritative participants aka permissioned network. Second, gossip propagation is slow with a latency of $O(n)$ for $n$ participants. Third, it remains unclear whether their consensus and final event ordering algorithms are truly asynchronous BFT.

In addition, there is only a few research work studying Proof-of-Stake in DAG-based consensus protocols. One example is ~\cite{hashgraph16}, but it is briefly mentioned in that paper.

Hence, in this paper, we are interested in a new approach to address the aforementioned issues in existing pBFT and DAG-based approaches, and the new approach leverages participant's stake to improve DAG-based consensus decisions. Specifically, we propose a new consensus algorithm that addresses the following questions:
\begin{itemize}
	\item Can we achieve a public Proof-Of-Stake + DAG-based consensus protocol?
	\item Can pBFT be achieved when time-to-finality (TTF) for transactions is kept closed to 1 second?
	\item Can we reach local consensus in a $k$-cluster faster for some $k$?
	\item Can we achieve faster event propagation such as using a subset broadcast?
	\item Can continuous common knowledge be used for consensus decisions with high probability?
\end{itemize}

\subsection{Lachesis protocol: Proof-Of-Stake DAG aBFT}
In this paper, we introduce our Lachesis protocol, denoted by $\mathcal{L}$ to reach faster consensus using topological ordering of events for an asynchronous non-deterministic distributed system that guarantees pBFT with deterministic finality.

The core idea of Lachesis is the OPERA DAG, which is a block DAG.
Nodes generate and propagate event blocks asynchronously and the Lachesis algorithm achieves consensus by confirming how many nodes know the event blocks using the OPERA DAG. 
In Lachesis protocol, a node can create a new event block, which has a set of 2 to $k$ parents. OPERA DAG is used to compute special event blocks, such as Root, Clotho, and Atropos. 
The Main chain consists of ordered Atropos event blocks. It can maintain reliable information between event blocks. The OPERA DAG and Main chain are updated frequently with newly generated event blocks and can respond strongly to attack situations such as forking and parasite attack.

{\bf PoS + DAG}: We introduce StakeDag protocol~\cite{stakedag} and StairDag protocol~\cite{stairdag} presents a general model that integrates Proof of Stake model into DAG-based consensus protocol Lachesis. Both generate blocks asynchronously to build a \emph{weighted DAG} from Validator blocks. Consensus on a block is computed from the gained validating power of validators on the block. We use Lamport timestamp, Happened-before relation between event blocks, graph layering and hierarchical graphs on the weighted DAG, to achieve deterministic topological ordering of finalized event blocks in an asynchronous leaderless DAG-based system.

\begin{figure}[h] \centering
	\includegraphics[width=0.9\columnwidth]{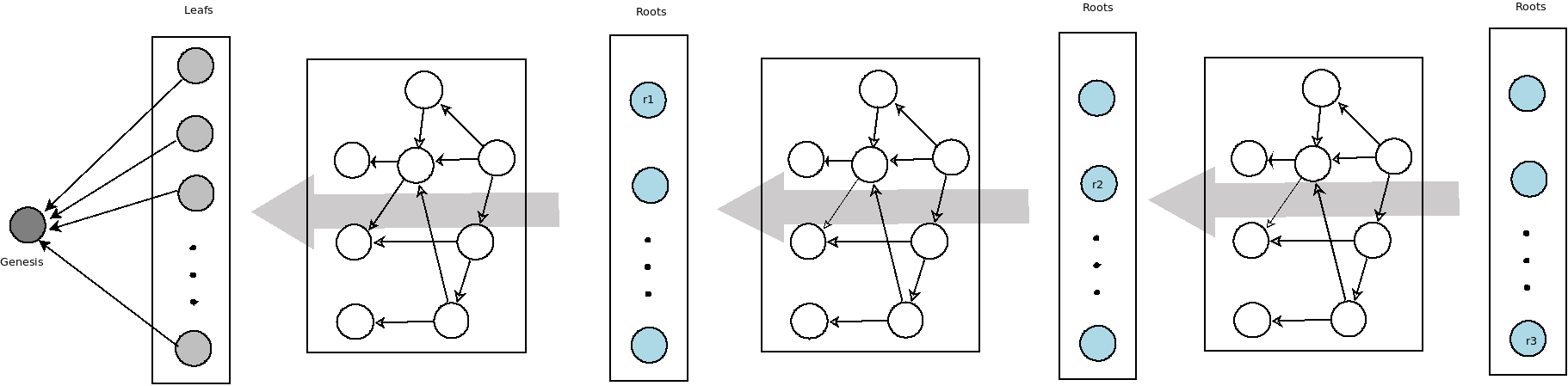}
	\caption{Consensus Method through Path Search in a DAG (combines chain with consensus process of pBFT)}
	\label{fig:pBFTtoPath}
\end{figure}
Figure~\ref{fig:pBFTtoPath} illustrates how consensus is reached through the path search in the OPERA DAG. Leaf set, denoted by $R_{s0}$, consists of the first event blocks created by individual participant nodes. Let $R_{si}$ be the $i$-th root set.
A root $r_1$ in $R_{s1}$ can reach more than a quorum in the leaf set. A quorum is 2/3W, where W is the total validating power in a weighted DAG, or 2/3n where n is the number nodes in an unweighted DAG. A root $r_3$ in a later root set may confirm some root $r'$ in $R_{s1}$ and  the subgraph of $r'$. Lachesis protocol reaches consensus in this manner and is similar to the proof of pBFT approach.

{\bf Dynamic participation}: Lachesis protocol supports \emph{dynamic participation} so that all participants can join the network~\cite{fantom18}.

{\bf Layering}: Lachesis protocol leverages the concepts of graph layering and hierarchical graphs on the DAG, as introduced in our ONLAY framework~\cite{onlay19}. Assigned layers are used to achieve deterministic topological ordering of finalized event blocks in an asynchronous leaderless DAG-based system.

The main concepts of Lachesis protocol are given as follows:
\begin{itemize}
	\item Event block: Nodes can create event blocks. Event block includes the signature, generation time, transaction history, and reference to parent event blocks.
	\item Happened-before: is the relationship between nodes which have event blocks. If there is a path from an event block $x$ to $y$, then $x$ Happened-before $y$. ``$x$ Happened-before $y$" means that the node creating $y$ knows event block $x$.
	\item Lamport timestamp: For topological ordering, Lamport timestamp algorithm uses the happened-before relation to determine a partial order of the whole event block based on logical clocks.
	\item Stake: This corresponds to the amount of tokens each node possesses in their deposit. This value decides the validating power a node can have.
	\item User node: A user node has a small amount stake (e.g., containing 1 token).
	\item Validator node: A validator node has large amount of stake ($\geq$ 2 tokens).
	\item Validation score: Each event block has a validation score, which is the sum of the weights of the roots that are reachable from the block.	
	\item OPERA DAG: is the local view of the DAG held by each node, this local view is used to identify topological ordering, select Clotho, and create time consensus through Atropos selection.
	\item S-OPERA DAG: is the local view of the weighted Directed Acyclic Graph (DAG) held by each node. This local view is used to determine consensus.
	\item Root: An event block is called a \emph{root} if either (1) it is the first event block of a node, or (2) it can reach more than $2/3$ of the network's validating power from other roots. A \emph{root set} $R_s$ contains all the roots of a frame. A \emph{frame} $f$ is a natural number assigned to Root sets and its dependent event blocks.
	\item Clotho: A Clotho is a root at layer $i$ that is known by a root of a higher frame ($i$ + 1), and which in turns is known by another root in a higher frame ($i$ +2).
	\item Atropos: is a Clotho assigned with a consensus time.
	\item Main chain: \stakedag's Main chain is a list of Atropos blocks and the subgraphs reachable from those Atropos blocks.
\end{itemize}

\subsection{Contributions}

In summary, this paper makes the following contributions:
\begin{itemize}
	\item We propose a novel consensus protocol, namely Lachesis, which is Proof-Of-Stake DAG-based, aiming for practical aBFT distributed ledger.
	\item Our Lachesis protocol uses Lamport timestamp and happened-before relation on the DAG, for faster consensus. Our consensus algorithms for selection of roots, Clothos and Atropos are deterministic and scalable.
	\item Lachesis protocol allows \emph{dynamic participation} for all participants to join the network.
	\item Our protocol leverages \emph{epochs} or \emph{checkpoints} to significantly improve storage usage and achieve faster consensus computation.
	\item We define a formal model using continuous consistent cuts of a local view to achieve consensus via layer assignment. Our work is the first that give a detailed model, formal semantics and proofs for an aBFT PoS DAG-based consensus protocol. 
\end{itemize}

\subsection{Paper structure}

The rest of this paper is organized as follows. 
Section~\ref{se:related} gives the related work. 
Section~\ref{se:lachesis} presents our model of Lachesis protocol that is Proof of Stake DAG-based. 
Section~\ref{se:staking} 
introduces a staking model that is used for our \stair\ consensus protocol.
Section~\ref{se:network} describes the Opera network's overall architecture, node roles and running services to support Fantom's ecosystem. We also give a brief performance analysis and some statistics of the Opera network.
Section~\ref{se:discuss} gives some discussions about our Proof of Stake \stair\ protocol, such as fairness and security.
Section~\ref{se:con} concludes.

\section{Related work}\label{se:related}

\subsection{An overview of Blockchains}

A blockchain is a type of distributed ledger technology (DLT) to build record-keeping system in which its users possess a copy of the ledger. The system stores transactions into blocks that are linked together. The blocks and the resulting chain are immutable and therefore serving as a proof of existence of a transaction. 
In recent years, the blockchain technologies have seen a widespread interest, with applications across many sectors such as
finance, energy, public services and sharing platforms. 

For a public or permissionless blockchain, it has no central authority. Instead, consensus among users is paramount to guarantee the security and the sustainability of the system. For private, permissioned or consortium blockchains, one entity or a group of entities can control who sees, writes and modifies the data on it. 
We are interested in the decentralization of BCT and hence only public blockchains are considered in this section. 

In order to reach a global consensus on the blockchain, users must follow the rules set by the consensus protocol of the system.

\subsubsection{Consensus algorithms}

In a consensus algorithm, all participant nodes of a distributed network share transactions and agree integrity of the shared transactions~\cite{Lamport82}. It is equivalent to the proof of Byzantine fault tolerance in distributed database systems~\cite{randomized03, paxos01}. 
The Practical Byzantine Fault Tolerance (pBFT) allows all nodes to successfully reach an agreement for a block when a Byzantine node exists \cite{Castro99}. 

There have been extensive research in consensus algorithms. 
Proof of Work (PoW)~\cite{bitcoin08}, used in the original Nakamoto consensus protocol in Bitcoin, requires exhaustive computational work from participants for block generation. Proof Of Stake (PoS)~\cite{ppcoin12,dpos14} uses participants' stakes for generating blocks.
Some consensus algorithms~\cite{algorand16, algorand17, sompolinsky2016spectre, PHANTOM08} have addressed to improve the consensus confirmation time and power consumption over blockchain-powered distributed ledges. These approaches utilize directed acyclic graphs (DAG)~\cite{dagcoin15, sompolinsky2016spectre, PHANTOM08, PARSEC18, conflux18} to facilitate consensus. 
Examples of DAG-based consensus algorithms include Tangle~\cite{tangle17}, Byteball~\cite{byteball16}, and Hashgraph~\cite{hashgraph16}.
Lachesis protocol~\cite{lachesis01} presents a general model of DAG-based consensus protocols.

\subsubsection{Proof of Work} 
Bitcoin was the first and most used BCT application. Proof of Work (PoW) is the consensus protocol introduced by the Bitcoin in 2008~\cite{bitcoin08}.  
In PoW protocol, it relies on user's computational power to solve a cryptographic puzzle that creates consensus and ensures the integrity of data stored in the chain. Nodes validate transactions in blocks (i.e. verify if sender has sufficient funds and is not double-spending) and competes with each other to solve the puzzle set by the protocol. The incentive for miners to join this mining process is two-fold: the first miner, who finds a solution, is rewarded (block reward) and gains all the transaction fees associated to the transactions. Nodes validate transactions in blocks with an incentive to gain block reward and transaction fees associated to the transactions. 

The key component in Bitcoin protocol and its successors is the PoW puzzle solving. The miner that finds it first can issue the next block and her work is rewarded in cryptocurrency.
PoW comes together with an enormous energy demand. 

From an abstract view, there are two properties of PoW blockchain:
\begin{itemize}
	\item Randomized leader election: The puzzle contest winner is elected to be the leader and hence has the right to issue the next block. The more computational power a miner has, the more likely it can be elected.
	\item Incentive structure: that keeps miners behaving honestly and extending the blockchain. In Bitcoin this is achieved by miners getting a block discovery reward and transaction fees from users. In contrast, subverting the protocol would reduce the trust of the currency and would lead to price loss. Hence, a miner would end up undermining the currency that she itself collects.
\end{itemize}

Based on those two characteristics, several  alternatives to Proof-of-Work have been proposed.

\subsubsection{Proof of Stake}

\emph{Proof of Stake}(PoS) is an alternative to PoW for blockchains. 
PoS relies on a lottery like system to select the next leader or block submitter.  
Instead of using puzzle solving, PoS is more of a lottery system.  Each node has a certain amount of stake in a blockchain. Stake can be the amount of currency, or the age of the coin that a miner holds.
In PoS, the leader or block submitter is randomly elected with a probability proportional to the amount of stake it owns in the system. The elected participant can issue the next block and then is rewarded with transaction fees from participants whose data are included. The more stake a party has, the more likely it can be elected as a leader. Similarly to PoW, block issuing is rewarded with transaction fees from participants whose data are included. 
The underlying assumption is that: stakeholders are incentivized to act in its interests, so as to preserve the system.

There are two major types of PoS. The first type is chain-based PoS~\cite{pass2017fruitchains}, which uses chain of blocks like in PoW, but stakeholders are randomly selected based on their stake to create new blocks. This includes Peercoin~\cite{king2012ppcoin}, Blackcoin~\cite{vasin2014blackcoin}, and Iddo Bentov’s work~\cite{bentov2016}, just to name a few. The second type is BFT-based PoS that is based on BFT consensus algorithms such as pBFT~\cite{Castro99}. Proof of stake using BFT was first introduced by Tendermint~\cite{kwon2014tendermint}, and has attracted more research~\cite{algorand16}. 
Ethereum's Parity project have investigated to migrate into a PoS blockchain~\cite{buterin2018}. 

PoS is more advantageous than PoW because it only requires cheap hardware and every validator can submit blocks. PoS approach reduces the energy demand and is more secure than PoW.

Every validator can submit blocks and the likelihood of acceptance is proportional to the \% of network weight (i.e., total amount of tokens being staked) they possess. Thus, to secure the blockchain, nodes need the actual native token of that blockchain. To acquire the native tokens, one has to purchase or earn them with staking rewards. Generally, gaining 51\% of a network’s stake is much harder than renting computation.

{\bf Security} Although PoS approach reduces the energy demand, new issues arise that were not present in PoW-based blockchains. These issues are shown as follows:
\begin{itemize}
	\item {\it Grinding attack:} Malicious nodes can play their bias in the election process to gain more rewards or to double spend their money.
	\item {\it Nothing at stake attack:} In PoS,  constructing alternative chains becomes easier. A node in PoS seemingly does not lose anything by also mining on an alternative chain, whereas it would lose CPU time if working on an alternative chain in PoW.
\end{itemize}

{\bf Delegated Proof of Stake}
To tackle the above issues in PoS, 
\emph{Delegated Proof of Stake} (DPoS) consensus protocols are introduced, such as Lisk, EOS~\cite{EOS}, Steem~\cite{Steem}, BitShares~\cite{bitshares} and Ark~\cite{Ark}. 
DPoS uses voting to reach consensus among nodes more efficiently by speeding up transactions and block creation. 
Users have different roles and have a incentive to behave honestly in their role.

In DPoS systems, users can \emph{vote} to select \emph{witnesses}, to whom they trust, to create blocks and validate transactions. For top tier witnesses that have earned most of the votes, they earn the right to validate transactions. Further, users can also \emph{delegate} their voting power to other users, whom they trust, to vote for witnesses on their behalf. 
In DPoS, votes are weighted based on the stake of each voter. A user with small stake amount can become a top tier witness, if it receives votes from users with large stakes.

Top witnesses are responsible for validating transactions and creating blocks, and then get fee rewards. Witnesses in the top tier can exclude certain transactions into the next block. But they cannot change the details of any transaction. There are a limited number of witnesses in the system.
A user can replace a top tier witness if s/he gets more votes or is more trusted. Users can also vote to remove a top tier witness who has lost their trust. Thus, the potential loss of income and reputation is the main incentive against malicious behavior in DPoS.

Users in DPoS systems also vote for a group of \emph{delegates}, who are trusted parties responsible for maintaining the network. Delegates are in charge of the governance and performance of the entire blockchain protocol. But the delegates cannot do transaction validation and block generation.
For example, they can propose to change block size, or the reward a witness can earn from validating a block. 
The proposed changes will be voted by the system's users.

DPoS brings various benefits:
(1) faster than traditional PoW and PoS Stake systems; 
(2) enhance security and integrity of the blockchains as each user has an incentive to perform their role honestly;
(3) normal hardware is sufficient to join the network and become a user, witness, or delegate;
(4) more energy efficient than PoW.

{\bf Leasing Proof Of Stake} Another type of widely known PoS is Leasing Proof Of Stake (LPoS). Like DPoS, LPoS allows users to vote for a delegate that will maintain the integrity of the system. Further, users in a LPoS system can lease out their coins and share the rewards gained by validating a block.

{\bf Proof of Authority} Another successor of PoS is \emph{Proof of Authority} (PoA). In PoA, the reputation of a validator acts as the stake~\cite{ProofofAuth}. PoA runs by a set of validators in a permissioned system. It gains higher throughput by reducing the number of messages sent between the validators. Reputation is difficult to regain once lost and thus is a better choice for ``stake''.

There are a number of surveys that give comprehensive details of PoW and PoS, such as~\cite{sheikh2018proof, panarello2018survey}.

Based on \stair~\cite{stairdag} and \stakedag~\cite{stakedag}, our Lachesis consensus protocos can be used for public permissionless asynchronous Proof of Stake systems. Participants are allowed to delegate their stake to a node to increase validating power of a node and share the validation rewards. Unlike existing DAG-based previous work, our Lachesis protocols are PoS+DAG approach.

\subsubsection{DAG-based approaches}

DAG-based approaches have currently emerged as a promising alternative to the PoW and PoS blockchains.
The notion of a \emph{DAG} (directed acyclic graph) was first coined in 2015 by DagCoin~\cite{dagcoin15}. Since then, DAG technology has been adopted in numerous systems, for example, ~\cite{dagcoin15, sompolinsky2016spectre, PHANTOM08, PARSEC18, conflux18}. Unlike a blockchain, DAG-based system facilitate consensus while achieving horizontal scalability. 

DAG technology has been adopted in numerous systems. This section will present the popular DAG-based approaches.
Examples of DAG-based approaches include Tangle~\cite{tangle17}, Byteball~\cite{byteball16}, 
Hashgraph~\cite{hashgraph16}, 
RaiBlocks~\cite{raiblock17},
Phantom~\cite{PHANTOM08},
Spectre~\cite{sompolinsky2016spectre},
Conflux~\cite{conflux18},
Parsec~\cite{PARSEC18} and Blockmania~\cite{Blockmania18}
.

Tangle is a DAG-based approach proposed by IOTA~\cite{tangle17}. 
Tangle uses PoW to defend against sybil and spam attacks. Good actors need to spend a considerable amount of computational power, but a bad actor has to spend increasing amounts of power for diminishing returns. Tips based on transaction weights are used to address the double spending and parasite attack. 

Byteball~\cite{byteball16} introduces an internal pay system called Bytes used in distributed database. Each storage unit is linked to previous earlier storage units. The consensus ordering is computed from a single main chain consisting of roots. Double spends are detected by a majority of roots in the chain.

Hashgraph~\cite{hashgraph16} introduces an asynchronous DAG-based approach in which each block is connected with its own ancestor. Nodes randomly communicate with each other about their known events. Famous blocks are computed by using \textit{see} and \text{strong see} relationship at each round. Block consensus is achieved with the quorum of more than 2/3 of the nodes.

RaiBlocks~\cite{raiblock17} was proposed to improve high fees and slow transaction processing. Consensus is obtained through the balance weighted vote on conflicting transactions. Each participating node manages its local data history. Block generation is carried similarly as the anti-spam tool of PoW. The protocol requires verification of the entire history of transactions when a new block is added.

Phantom~\cite{PHANTOM08} is a PoW based permissionless protocol that generalizes Nakamoto’s blockchain to a DAG. A parameter $k$ is used to adjust the tolerance level of the protocol to blocks that were created concurrently. The adjustment can accommodate higher throughput; thus avoids the security-scalability trade-off as in Satoshi’s protocol. A greedy algorithm is used on the DAG to distinguish between blocks by honest nodes and the others. It allows a robust total order of the blocks that is eventually agreed upon by all honest nodes.

Like PHANTOM, the GHOSTDAG protocol selects a $k$-cluster, which induces a colouring of the blocks as Blues (blocks in the selected cluster) and Reds (blocks outside the cluster). GHOSTDAG finds a cluster using a greedy algorithm, rather than looking for the largest $k$-cluster.

Spectre~\cite{sompolinsky2016spectre} uses DAG in a PoW-based protocol to tolerate from attacks with up to 50\% of the computational power. The protocol gives a high throughput and fast confirmation time. Sprectre protocol satisfies weaker properties in which the order between any two transactions can be decided from the transactions by honest users; whilst conventionally the order must be decided by all non-corrupt nodes. 

Conflux~\cite{conflux18} is a DAG-based Nakamoto consensus protocol. It optimistically processes concurrent blocks without discarding any forks. The  protocol achieves consensus on a total order of the blocks, which is decided by all participants. Conflux can tolerate up to half of the network as malicious while the BFT-based approaches can only tolerate up to one third of malicious nodes.

Parsec~\cite{PARSEC18} proposes a consensus algorithm in a randomly synchronous BFT network. It has no leaders, no round robin, no PoW and reaches eventual consensus with probability one. Parsec can reach consensus quicker than Hashgraph~\cite{hashgraph16}. The algorithm reaches 1/3-BFT consensus with very weak synchrony assumptions. Messages are delivered with random delays, with a finite delay in average.

Blockmania~\cite{Blockmania18}
achieves consensus with several advantages over the traditional pBFT protocol. In Blockmania, nodes in a quorum only emit blocks linking to other blocks, irrespective of the consensus state machine. The resulting DAG of blocks is used to ensure consensus safety, finality and liveliness.
It reduces the communication complexity to $O(N^2)$ even in the worse case, as compared to pBFT's complexity of $O(N^4)$.

In this paper, our \stakedag\ protocol is different from the previous work. We propose a general model of DAG-based consensus protocols, which uses Proof of Stake for asynchronous permissionless BFT systems. \stakedag\ protocol is based on our previous DAG-based protocols\cite{lachesis01,fantom18} to achieve asynchronous non-deterministic pBFT. The new $S_\phi$ protocol, which is based on our ONLAY framework~\cite{onlay19}, uses graph layering to achieve scalable, reliable consensus in a leaderless aBFT DAG.

\subsection{Lamport timestamps}

Lamport~\cite{lamport1978time} defines the "happened before" relation between any pair of events in a distributed system of machines. The happened before relation, denoted by $\rightarrow$, is defined without using physical clocks to give a partial ordering of events in the system. The relation "$\rightarrow$" satisfies the following three conditions: (1) If $b$ and $b'$ are events in the same process, and $b$ comes before $b'$, then $b \rightarrow b'$. (2) If $b$ is the sending of a message by one process and $b'$ is the receipt of the same message by another process, then $b \rightarrow b'$. (3) If $b \rightarrow b'$ and $b' \rightarrow b''$ then $b \rightarrow b''$. 
Two distinct events $b$ and $b'$ are said to be concurrent if $b \nrightarrow b'$ and $b' \nrightarrow b$.

The happened-before relation can be viewed as a causality effect: that $b \rightarrow  b'$ implies event $b$ may causally affect event $b'$. Two events are concurrent if neither can causally affect the other.

Lamport introduces logical clocks which is a way of assigning a number to an event. A clock $C_i$ for each process $P_i$ is a function which assigns a number $C_i(b)$ to any event $b \in P_i$. The entire system of blocks is represented by the function $C$ which assigns to any event $b$ the number $C(b)$, where $C(b) = C_j(b)$ if $b$ is an event in process $P_j$.
The Clock Condition states that for any events $b$, $b'$: if $b \rightarrow b'$ then $C(b)$ $<$ $C(b')$. 

The clocks must satisfy two conditions. First, each process $P_i$ increments $C_i$ between any two successive events. Second, we require that each message $m$ contains a timestamp $T_m$, which equals the time at which the message was sent. Upon receiving a message timestamped $T_m$, a process must advance its clock to be later than $T_m$. 


Given any arbitrary total ordering $\prec$ of the processes, the total ordering $\Rightarrow$ is defined as follows: if $a$ is an event in process $P_i$ and $b$ is an event in process $P_j$, then $b \Rightarrow b'$ if and only if either (i) $C_i(b) < C_j(b')$ or (ii) $C(b)= Cj(b')$ and $P_i \prec P_j$. The Clock Condition implies that if $b \rightarrow b'$ then $b \Rightarrow b'$.

\subsection{Concurrent common knowledge}\label{se:cck}

In the Concurrent common knowledge (CCK) paper ~\cite{cck92}, they define a model to reason about the concurrent common knowledge in asynchronous, distributed systems. A system is composed of a set of processes that can communicate only by sending messages along a fixed set of channels. The network is not necessarily completely connected. The system is asynchronous in the sense that there is no global clock in the system, the relative speeds of processes are independent, and the delivery time of messages is finite but unbounded.

A local state of a process is denoted by $s^j_i$. Actions are state transformers; an action is a function from local states to local states. An action can be either: a send(m) action where m is a message, a receive(m) action, and an internal action.
A local history, $h_i$, of process $i$, is a (possibly infinite) sequence of alternating local states—beginning with a distinguished initial state—and actions. We write such a sequence as follows: 
$h_i = s_i^0 \xrightarrow{ \alpha_i^1 } s_i^1 \xrightarrow{\alpha_i^2} s_i^2 \xrightarrow{\alpha_i^3} ...$
The notation of $s^j_i$ ($\alpha^j_i$) refers to the $j$-th state (action) in process $i$'s local history
An event is a tuple $\langle s , \alpha, s' \rangle$ consisting of a state, an action, and a state.
The $j$th event in process $i$'s history is $e^j_i$ denoting $\langle s^{j-1}_i , \alpha^j_i, s^j_{i} \rangle$.

An asynchronous system consists of the following sets.
\begin{enumerate}
	\item A set $P$ = \{1,...,$N$\} of process identifiers, where $N$ is the total number of processes in the system.
	\item $A$ set $C$ $\subseteq$ \{($i$,$j$) s.t. $i,j \in P$\} of channels. The occurrence of $(i,j)$ in $C$ indicates that process $i$ can send messages to process $j$.
	\item A set $H_i$ of possible local histories for each process $i$ in $P$.
	\item A set $A$ of asynchronous runs. Each asynchronous run is a vector of local histories, one per process, indexed by process identifiers. Thus, we use the notation
	$a = \langle h_1,h_2,h_3,...h_N \rangle$.
	Constraints on the set $A$ are described throughout this section.
	\item A set $M$ of messages. A message is a triple $\langle i,j,B \rangle$ where $i \in P$ is the sender of the message, $j \in P$ is the message recipient, and $B$ is the body of the message. $B$ can be either a special value (e.g. a tag to denote a special-purpose message), or some proposition about the run (e.g. “$i$ has reset variable $X$ to zero”), or both. We assume, for ease of exposition only, that messages are unique.
\end{enumerate}

The set of channels $C$ and our assumptions about their behavior induce two constraints on the runs in $A$. First, $i$ cannot send a message to $j$ unless $(i,j)$ is a channel. Second, if the reception of a message $m$ is in the run, then the sending of $m$ must also be in that run; this implies that the network cannot introduce spurious messages or alter messages.

The CCK model of an asynchronous system does not mention time. Events are ordered based on Lamport's happened-before relation. They use Lamport’s theory to describe global states of an asynchronous system. A global state of run $a$ is an $n$-vector of prefixes of local histories of $a$, one prefix per process. Happened-before relation can be used to define a consistent global state, often termed a consistent cut, as follows. 

\begin{defn}[Consistent cut] A consistent cut of a run is any global state such that if $e^x_i \rightarrow e^y_j$ and $e^y_j$ is in the global state, then $e^x_i$ is also in the global state.
\end{defn} 

A message chain of an asynchronous run is a sequence of messages $m_1$, $m_2$, $m_3$, $\dots$, such that, for all $i$, $receive(m_i)$ $\rightarrow$  $send(m_{i+1})$. Consequently,
$send(m_1)$ $\rightarrow$ $receive(m_1)$ $\rightarrow$ $send(m_2)$ $\rightarrow$ $receive(m_2)$ $\rightarrow$ $send(m_3)$ $\dots$.

\section{Lachesis: Proof-Of-Stake DAG-based Protocol}\label{se:lachesis}

This section presents the key concepts of our new PoS+DAG consensus protocol, namely Lachesis.

In BFT systems, a synchronous approach utilizes a broadcast voting and asks each node to vote on the validity of each block. Instead, we aim for an asynchronous system where we leverage the concepts of distributed common knowledge and network broadcast to achieve a local view with high probability of being a consistent global view.

Each node in Lachesis protocol receives transactions from clients and then batches them into an event block. The new event block is then communicated with other nodes through asynchronous event transmission. During communication, nodes share their own blocks as well as the ones they received from other nodes. Consequently, this spreads all information through the network.
The process is asynchronous and thus it can  increase throughput near linearly as nodes enter the network.
 
\subsection{OPERA DAG}

In Lachesis protocol, each (participant) node, which is a server (machine) of the distributed system, can create messages, send messages to and receive messages from other nodes. The communication between nodes is asynchronous. Each node stores a local block DAG, so-called OPERA DAG, which is Directed Acyclic Graph (DAG) of event blocks. A block has some edges to its parent event blocks. A node can create an event block after the node communicates the current status of its OPERA DAG with its peers. 

\dfnn{OPERA DAG}{The OPERA DAG is a graph structure stored on each node. The OPERA DAG consists of event blocks and references between them as edges.}

OPERA DAG is a DAG graph $G$=($V$,$E$) consisting of $V$ vertices and $E$ edges. Each vertex $v_i \in V$ is an event block. An edge ($v_i$,$v_j$) $\in E$ refers to a hashing reference from $v_i$ to $v_j$; that is, $v_i \erefz v_j$. We extend with $G$=($V$,$E$,$\top$,$\bot$), where $\top$ is a pseudo vertex, called \emph{top}, which is the parent of all top event blocks, and 
$\bot$ is a pseudo vertex, called bottom, which is the child of all leaf event blocks.
With the pseudo vertices, we have $\bot$ happened-before all event blocks. Also all event blocks happened-before $\top$. That is, for all event $v_i$, $\bot \hbefore v_i$ and $v_i \hbefore \top$.

Figure~\ref{fig:operachain} shows an example of an OPERA DAG (DAG) constructed through the Lachesis protocol. Event blocks are represented by circles.

\begin{figure}[H] \centering
	\includegraphics[width=.8\columnwidth]{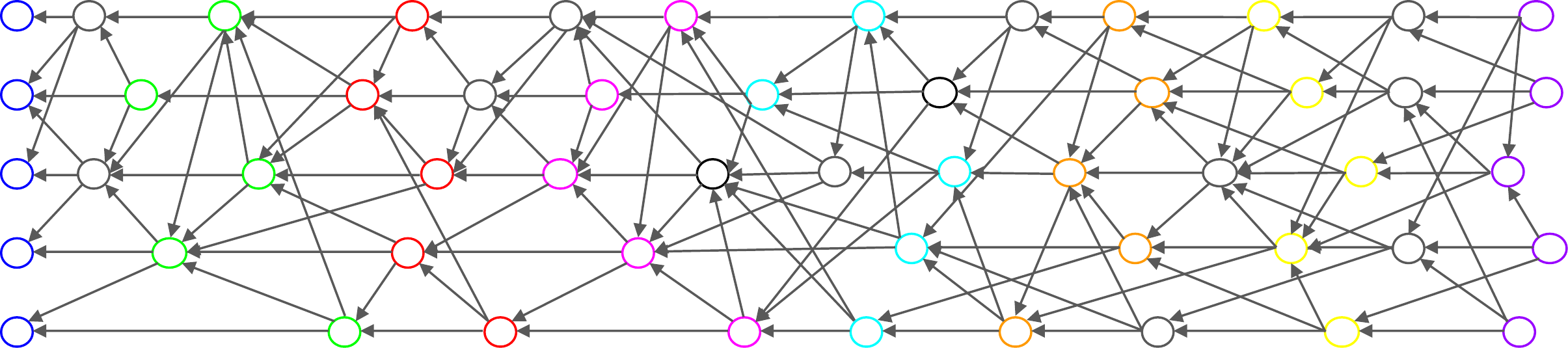}
	\caption{An Example of OPERA DAG}
	\label{fig:operachain}
\end{figure}

The local DAG is updated quickly as each node creates and synchronizes events with each other. 
For high speed transaction processing, we consider a model of DAG as \emph{DAG streams}.
Event blocks are assumed to arrive at very high speed, and they are asynchronous. 
Let $G$=($V$,$E$) be the current OPERA DAG and $G'$=($V'$,$E'$)
denote the \emph{diff graph}, which consists of the changes to $G$ at a time, either at event block creation or arrival. The vertex sets $V$ and $V'$ are disjoint; similar to the edge sets $E$ and $E'$. At each graph update, the updated OPERA DAG becomes $G_{new}$=($V \cup V'$, $E \cup E'$).

Each node uses a local view of OPERA DAG to identify Root, Clotho and Atropos vertices, and to determine topological ordering of the event blocks.

\subsection{PoS Definitions}

Inspired from the success of PoS~\cite{pass2017fruitchains,king2012ppcoin,vasin2014blackcoin,bentov2016,kwon2014tendermint,algorand16}, we introduce a general model of \stakedag\ protocols that are Proof of Stake DAG-based for trustless systems~\cite{stakedag}.
In this general model, the stakes of the participants are paramount to a trust indicator.

A stake-based consensus protocol consists of nodes that vary in their amount of stake. Based on our generic Lachesis protocol, any participant can join the network. Participants can increase their impact as well as their contributions over the Fantom network by increasing the stake amount (FTM tokens) they possess.
Validating power is based on the amount of FTM tokens a node possesses. The validating power is used for online validation of new event blocks to maintain the validity of blocks and the underlying DAG.

\subsubsection{Stake}
Each participant node of the system has an account.
The \emph{stake} of a participant is defined based on their account balance. Account balance is the number of tokens that was purchased, accumulated and/or delegated from other account.
Each participant has an amount of stake $w_i$, which is a non-negative integer.

Each participant with a positive balance can join as part of the system. A user has a stake of 1, whereas a validator has a stake of $w_i > 1$. 
The number of stakes is the number of tokens that they can prove that they possess. 
In our model, a participant, whose has a zero balance, cannot participate in the network, either for block creation nor block validation.

\subsubsection{Validating Power and Block Validation}

In the base Lachesis protocol, every node can submit new transactions, which are batched in a new event block, and it can communicate its new (own) event blocks with peers. Blocks are then validated by all nodes, which have the same validating power. 

Here, we present a general model that distinguishes Lachesis participants by their validating power.
In this model, every node $n_i$ in Lachesis has a stake value of $w_i$. This value is then used to determine their validating power to validate event blocks.

The base Lachesis protocol contains only nodes of the same weights; hence, it is equivalent to say that all nodes have a unit weight ($w_i$=1). In contrast, the full version of Lachesis protocol is more general in which nodes can have different validating power $w_i$.

In order to reach asynchronous DAG-based consensus, validators in Lachesis protocol are required to create event block (with or without transactions) to indicate that which block (and all of its ancestors) they have validated.

\subsubsection{Stake-based Validation}
In Lachesis, a node must validate its current block and the received ones before it attempts to create (or add) a new event block into its local DAG. A node must validate its (own) new event block before it communicates the new block to other nodes.

A root is an important block that is used to compute the final consensus of the event blocks of the DAG. Lachesis protocol take the stakes of participants into account to compute the consensus of blocks. The weight of a block is the sum of the validating power of the nodes whose roots can be reached from the block. 

\subsubsection{Validation Score}
The \emph{validation score} of a block is the total of all validating powers that the block has gained.
The validation score of a block $v_i$ $\in$ $G$ is denoted by $s(v_i)$.
If the validation score of a block is greater than 2/3 of the total validating power, the block becomes a root. The weight of a root $r_i$ is denoted by $w(r_i)$, which is the weight $w_j$ of the creator node $j$ of the $r_i$. 
When an event block becomes a root, its score is the validating power of the root's creator.

\subsection{S-OPERA DAG: A weighted DAG model}

Like \stakedag\ protocol, Lachesis protocol uses a DAG-based structure, namely \emph{S-OPERA DAG}, which is a weighted DAG $G$=($V$,$E$, $w$), where $V$ is the set of event blocks, $E$ is the set of edges between the event blocks, and $w$ is a weight mapping to associates a vertex $v_i$ with its weight $w(v_i)$. Each vertex (or event block) is associated with a weight (validation score). Event block has a validating score, which is the total weights of the roots reachable from it.
A root vertex has a \emph{weight} equal to the validating power of the root's creator node.

\begin{figure}[ht]
	\centering
	\includegraphics[width=0.4\linewidth]{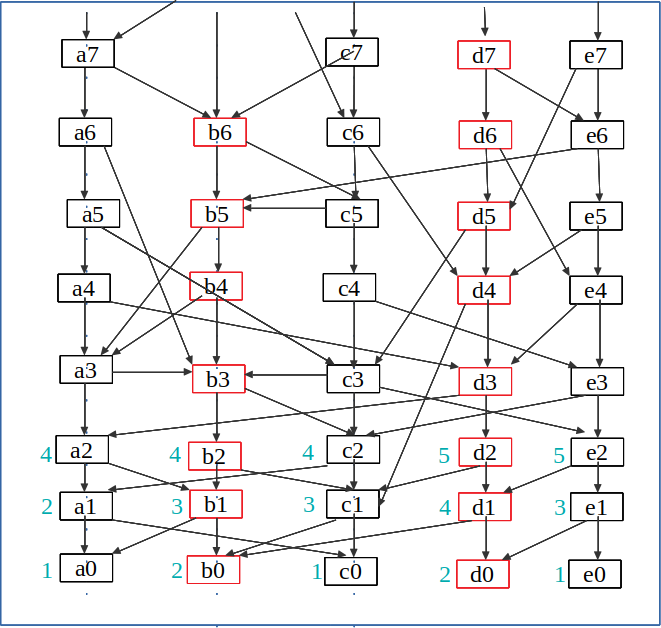}
	\caption{An example of an S-OPERA DAG in Lachesis. The validation scores of some selected blocks are shown.}
	\label{fig:stakedag-validatorscore}
\end{figure}
Figure \ref{fig:stakedag-validatorscore} depicts an S-OPERA DAG obtained from the DAG.
In this example, validators have validating power of 2, while users have validating power of 1. Blocks created by validators are highlighted in red color. First few event blocks are marked with their validating power. Leaf event blocks are special as each of them has a validation score equal to their creator's validating power.

\dfnn{S-OPERA DAG}{In Lachesis, a S-OPERA DAG is weighted DAG stored on each node.}

Let $\mathcal{W}$ be the total validating power of all nodes. 
For consensus, the algorithm examines whether an event block has a validation score of at least a quorum of $2\mathcal{W}/3$, which means the event block has been validated by more than two-thirds of total validating power in the S-OPERA DAG.

\subsection{Overall Framework}
\begin{figure}[ht]
	\centering
	\includegraphics[width=0.8\linewidth]{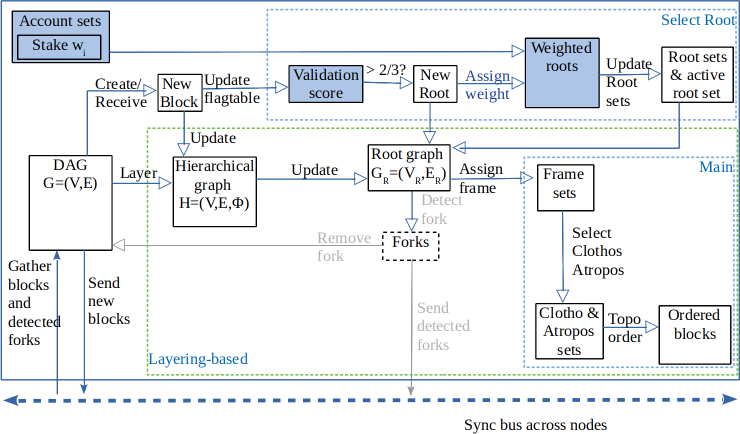}
	\caption{A General Framework of \stakedag\ Protocol}
	\label{fig:stakedagframework}
\end{figure}
Figure~\ref{fig:stakedagframework} shows a general framework of Lachesis protocol. Each node contains a DAG consisting of event blocks. In each node, the information of accounts and their stakes are stored. The main steps in a PoS DAG-based consensus protocol include (a) block creation, (b) computing validation score, (c) selecting roots and updating the root sets, (d) assigning weights to new roots, (e) decide frames, (f) decide Clothos/Atropos, and (f) order the final blocks.

\subsection{Quorum}
To become a root, an event block requires more than 2/3 of validating power (i.e., 2W/3) rather than 2n/3 in unweighted version of Lachesis.
For simplicity, we denote quorum $Q$ to be $2W/3 + 1$ for weighted DAG. In the original unweighted version, quorum $Q$ is $2n/3 +1$ for unweighted DAG.

\subsection{Event Block Creation}

Each event block has at most $k$ references to other event blocks using their hash values. One of the references must be a self-ref (or self-parent)  that references to an event block of the same node. The other references, which are also known as other-parent or other-ref references, are the top event blocks of other nodes.

Prior to creating a new event block, the node will first validate its current block and the selected top event block(s) from other node(s). 
With cryptographic hashing, an event block can only be created or added if the self-ref and other-ref event blocks exist.

Nodes with a higher stake value $w_i$ have a higher chance to create a new event block than other nodes with a lower stake.

\subsection{Event Block Structure}

An event (or an event block) is a vertex in the DAG, which stores the structure of consensus messages. Each event has a number of links to past events (parents). The links are the graph edges. Sometimes, "event" and "event block" can be used interchangeably, but they are different from "block", which means a finalized block in the blockchain.

Each event is a signed consensus message by a validator sent to the network, which guarantees the following that:
"The creator of this event has observed and validated the past events (and their parents), including the parents of this event. The event's creator has already validated the transactions list contained in this event."

The structure of an event block in Lachesis is as follows:

\begin{itemize}
	\item event.Epoch: epoch number (see epoch). Not less than 1.
	\item event.Seq: sequence number. Equal to self-parent's seq + 1, if no self-parent, then 1.
	\item event.Frame: frame number (see consensus). Not less than 1.
	\item event.Creator: ID of validator which created event.
	\item event.PrevEpochHash: hash of finalized state of previous epoch.
	\item event.Parents: list of parents (graph edges). May be empty. If Seq > 1, then first element is self-parent.
	\item event.GasPowerLeft: amount of not spent validator's gas power for each gas power window, after connection of this event.
	\item event.GasPowerUsed: amount of spent validator's gas power in this event.
	\item event.Lamport: Lamport timestamp. If parents list is not empty, then Lamport timestamp of an event v, denoted by LS(v) is max{LS(u) | u is a parent of v} + 1. Otherwise, LS(v) is 1
	\item event.CreationTime: a UnixNano timestamp that is specified by the creator of event. Cannot be lower than creation time of self-parent (if self-parent exists). Can be too-far-in-future, or too-far-in-past.
	\item event.MedianTime: a UnixNano timestamp, which is a weighted median of highest observed events (their CreationTime) from each validator. Events from the cheater validators are not counted. This timestamp is used to protect against "too-far-in-future" and "too-far-in-past".
	\item event.TxHash: Merkle tree root of event transactions.
	\item event.Transactions: list of originated transactions.
	\item event.Sig: ECDSA256 secp256k1 validator's signature in the R/S format (64 bytes).
	
\end{itemize}

Each event has two timestamps. Event's creationTime is the created timestamp  an assigned by the creator of event block.  Event's MedianTime is a weighted median of highest observed events (their CreationTime) from each validator. Events from the cheater validators aren't counted.

\subsection{Root Selection}\label{se:rootsel}

Root events get consensus when 2/3 validating power is reached.
When a new root is found, 'Assign weight' step will assign a weight to the root.
The root's weight is set to be equal to  the validating power $w_i$ of its creator.
The validation score of a root is unchanged.
Figure~\ref{fig:stakedagframework} depicts the process of selecting a new root in Lachesis. A node can create a new event block or receive a new block from other nodes.

To guarantee that root can be selected correctly even in case a fork exists in the DAG, we introduce $forklessCause$ concept.
The relation $forklessCause(A,B)$ denote that event $x$ is forkless caused by event $y$. It means, in the subgraph of event $x$, $x$ does not observe any forks from $y$’s creator and at least QUORUM non-cheater validators have observed event $y$.

An event is called a root, if it is \emph{forkless causes} by QUORUM roots of previous frame. The root starts a new frame. For every event, its frame number is no smaller than self-parent's frame number. The lowest possible frame number is 1. Relation $forklessCause$ is a stricter version of \texttt{happened-before} relation, i.e. if $y$ forkless causes $x$, then $y$ is happened before A. 
Yet, unlike \texttt{happened-before} relation, \texttt{forklessCause} is not transitive because it returns false if fork is observed. So if $x$ forkless causes $y$, and $y$ forkless causes $z$, it does not mean that $x$ forkless causes $z$. If there's no forks in the graph, then the \texttt{forklessCause} is always transitive.

\subsection{Clotho Selection}\label{se:Clotho-sel}

Next, we present our algorithm to select Clotho from roots. A Clotho is a root that is known by more than a quorum $Q$ of nodes, and there is another root that knows this information. In order for a root $r_i$ in frame $f_i$ to become a Clotho, $r_i$ must be reach by some root $r_{i+3}$ in frame $f_{i+3}$. This condition ensures an equivalence to a proof of practical BFT.


\begin{algorithm}
	\caption{Clotho Selection}\label{al:selClotho}
	\begin{algorithmic}[1]
		\Procedure{DecideClotho}{}
		\For{$x$ range roots: lastDecidedFrame + 1 to $up$}
			\For{$y$ range roots: $x.frame$ + 1 to $up$}
				\State // y is a root which votes, x is a voting subject root
				\State round $\leftarrow$ y.frame - x.frame
				\If{ round == 1}
					\State // y votes positively if x forkless causes y
					\State y.vote[x] $\leftarrow$ forklessCause(x, y)
				\Else
					\State prevVoters $\leftarrow$ getRoots(y.frame-1) // contains at least QUORUM events, from definition of root
					
					\State yesVotes $\leftarrow$ 0
					\State noVotes $\leftarrow$ 0
					\For{prevRoot range prevRoots}
						\If{forklessCause(prevRoot, y) == TRUE} 
							\State // count number of positive and negative votes for x from roots which forkless cause y
							\If{ prevRoot.vote[x] == TRUE}
								\State yesVotes $\leftarrow$ yesVotes + prevVoter's STAKE
							\Else
								\State noVotes $\leftarrow$ noVotes + prevVoter's STAKE
							\EndIf
						\EndIf	
					\EndFor
					// y votes likes a majority of roots from previous frame which forkless cause y				
					\State y.vote[x] $\leftarrow$ (yesVotes - noVotes) $>=$ 0
					\If{yesVotes $>=$ QUORUM}
						\State x.candidate $\leftarrow$ TRUE // decided as a candidate for Atropos
						\State decidedRoots[x.validator] $\leftarrow$ x
						\State break
					\EndIf
					\If{noVotes $>=$ QUORUM}
						\State x.candidate $\leftarrow$ FALSE // decided as a non-candidate for Atropos
						\State decidedRoots[x.validator] $\leftarrow$ x
					 	\State break
					\EndIf
				\EndIf
			\EndFor
		\EndFor
		\EndProcedure
	\end{algorithmic}
\end{algorithm}

Algorithm~\ref{al:selClotho} shows a pseudo code to select Clothos from the current set of roots. For every root $x$ from last decided frame + 1, it then loops for every root $y$ from $x$.frame + 1, and executes the inner loop body. Coloring of the graph is performed such that each root is assigned with one of the colors: \texttt{IS-CLOTHO} (decided as candidate for Clotho), \texttt{IS-NOT-CLOTHO} (decided as not a candidate for Clotho), or a temporary state \texttt{UNDECIDED} (not decided). Once a root is assigned a color by a node, the same color is assigned to the root by all other honest nodes, unless more than 1/3W are Byzantine. This is the core property of Lachesis consensus algorithm.

\begin{itemize}
	\item If a root has received $w$ $\geq$ QUORUM votes for a color, then for any other root on next frame, its \texttt{prevRoots} (contains at least QUORUM events, from definition of root) will overlap with V, and resulting overlapping will contain more than 1/2 of the \texttt{prevVoters}. It iscrucial that QUORUM is $2/3W+1$, not just 2/3W. The condition implies that ALL the roots on next frame will vote for the same color unanimously.
	\item A malicious validator may have multiple roots on the same frame, as they may make some forks. The \texttt{forklessCause} relation and the condition that no more than $1/3W$ are Byzantine will guarantee that no more than one of these roots may be voted \texttt{YES}. The algorithm assigns a color only in round $>=$ 2, so \texttt{prevRoots} will be "filtered" by \texttt{forklessCause} relation, to prevent deciding differently due to forks.
	\item If a root has received at least one \texttt{YES} vote, then the root is guaranteed to be included in the graph. 
	\item A root event $x$ from an offline validator will be decided \texttt{NO} by all. Thus,  $x$ is considered as non-existent event and cannot be assigned \texttt{IS-CLOTHO} color, and hence it cannot become Atropos.
	\item Technically, the election may take many rounds. Practically, however, the election ends mostly in 2nd or 3rd round.
\end{itemize}

Unlike HashGraph's \texttt{decideFame} function in page 13~\cite{hashgraph16}, our Lachesis algorithm has no coin round and is completely deterministic.

A Clotho is an Atropos candidate that has not been assigned a timestamp.

\subsection{Atropos Selection}\label{se:Atropos-sel}

From the array of roots that are decided as candidate, there are several ways to sort these Atropos candidates. Atropos candidates can be sorted based on the Lamport timestamp, layering information, validator's stake, validator's id, and the Atropos's id (see our ONLAY paper~\cite{onlay19}).

Here, we present an algorithm that sorts these Atropos candidates using validator's stake and validator's id. Algorithm~\ref{al:acs} shows the pseudo code for Atropos selection. The algorithm sorts the decided roots based on their validator's information. Validators are sorted based on their stake amount first, and then by their Id. The algorithms requires a sorted list of validators. The sorting is $O(n.log(n))$, where $n$ is the number of roots to process.

\begin{algorithm} [H]
	\caption{Atropos Selection}\label{al:acs}
	\begin{algorithmic}[1]		
		\Function{DecideAtropos}{}
		\For{validator: range sorted validators}
			\State root = decidedRoots[validator]
			\If{root == NULL} // not decided 
				\State \Return nil // frame isn't decided yet
			\EndIf
			\If{root.candidate == TRUE}
				\State \Return root // found Atropos
			\EndIf
			\If{root.candidate == FALSE}
				\State continue
			\EndIf
		\EndFor
		\EndFunction		
	\end{algorithmic}
\end{algorithm}

Nodes in Lachesis reach consensus about Atropos selection and Atropos timestamp without additional communication (i.e., exchanging candidate time) with each other. Through the Lachesis protocol, the OPERA DAGs of all nodes are ``consistent``. This allows each node to know the candidate time of other nodes based on its OPERA DAG and reach a consensus agreement. Our proof is shown in Section~\ref{se:proof}.

Atropos event block, once decided, is final and non-modifiable. All event blocks can be in the subgraph of an Atropos guarantee finality.

\subsection{Atropos timestamp}

Validators are sorted by their stake amount and then by their ID. For roots from different validators, they are sorted based on the ordering of their validators.
Atropos is chosen as a first candidate root in the sorted order.
Notice that, an Atropos may be decided even though not all the roots on a frame are decided. Our Lachesis algorithm allows to decide frames much earlier when there's a lot of validators, and hence reduce TTF (time to finality).

\subsection{Main chain}

The \emph{Main-chain} is an append-only list of blocks that caches the final consensus ordering of the finalized Atropos blocks.
The local hashing chain is useful to improve path search to quickly determine the closest root to an event block. 
Each participant has an own copy of the Main chain and can search consensus position of its own event blocks from the nearest Atropos.
The chain provides quick access to the previous transaction history to efficiently process new coming event blocks.
After a new Atropos is determined, a topological ordering is computed over all event blocks which are yet ordered.
After the topological ordering is computed over all event block, Atropos blocks are determined and form the Main chain.

\subsection{Topological ordering of events}\label{se:toposort}

Every node has a physical clock, and each event block is assigned with a physical time (creation time) from the event's creator. 
Each event block is also assigned a Lamport timestamp. For consensus, Lachesis protocol relies on a logical clock (Lamport timestamp) for each node. 

We use \textit{Lamport timestamp} \cite{lamport1978time} to determine the time ordering between event blocks in a asynchronous distributed system.
The Lamport timestamps algorithm is as follows: (1) Each node increments its count value before creating an event block; (2) When sending a message, sender includes its count value. Receiver should consider the count sent in sender’s message is received and increments its count value, (3) If current counter is less than or equal to the received count value from another node, then the count value of the recipient is updated; (4) If current counter is greater than the received count value from another node, then the current count value is updated.

\subsection{Blocks}\label{se:blocks}

Next, we present our approach to ordering event blocks under Atropos event block that reaches finality. 

Once a new Atropos is decided, a new block will be created. The new block contains all the events in the Atropos's subgraph; it only includes events, which were yet included in the subgraphs of any previous Atroposes.
The events in the new block are ordered by Lamport time, then by event's hash. Ordering by Lamport time ensures that parents are ordered before their children. Our approach gives a deterministic final ordering of the events in the block.

Figure~\ref{fig:crust} illustrates the steps to make blocks from decided Atroposes. In the example, Clotho/Atropos vertices are colored red. For each Atropos vertex, we compute the subgraph under it. Atropos vertices can be selected by using of several criteria (see ONLAY paper~\cite{onlay19}). 
Subgraphs under Atropos vertices are shown in different colors. Atropos event blocks are processed in the order from lower layer to high layer. 

\begin{figure}[ht]
	\centering
	\includegraphics[width=0.3\linewidth]{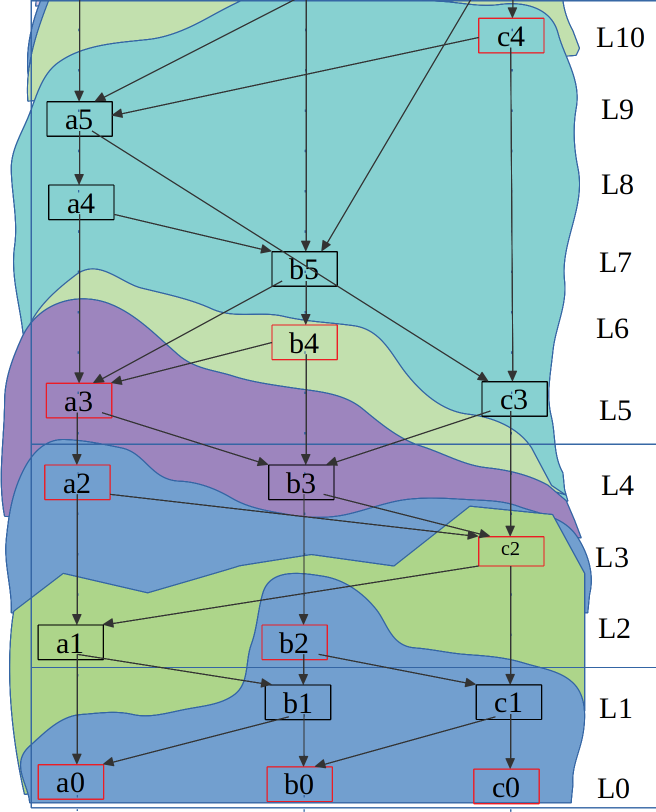}
	\caption{Peeling the subgraphs in topological sorting order}
	\label{fig:crust}
\end{figure}

We present an algorithm for topological ordering of Atropos event blocks (introduced in our ONLAY paper~\cite{onlay19}). Algorithm~\ref{algo:topoordering} first orders the Atropos candidates using \texttt{SortByLayer} function, which sorts vertices based on their layer, Lamport timestamp and then hash information of the event blocks. Second, we process every Atropos in the sorted order, and then compute the subgraph $G[a] = (V_a, E_a)$ under each Atropos. The set of vertices $V_u$ contains vertices from $V_a$ that are not yet processed in $U$.
We then apply $SorByLayer$ to order vertices in $V_u$ and then append the ordered result into the final ordered list $S$.

\begin{algorithm}[H]
	\caption{TopoSort}\label{algo:topoordering}
	\begin{algorithmic}[1]
		\State Require: OPERA DAG $H$, $\phi$, $\phi_F$.
		\Function{TopoSort}{$A$}
		\State $S \leftarrow$ empty list
		\State $U \leftarrow \emptyset$
		\State $Q \leftarrow$ SortByLayer($A$)
		\For{Atropos $a \in Q$}
		\State Compute subgraph $G[a] = (V_a, E_a)$
		\State $V_u \leftarrow V_a \setminus U$
		\State $T$ $\leftarrow$ SortByLayer($V_u$)
		\State Append $T$ into the end of ordered list $S$.
		\EndFor
		\EndFunction
		
		\Function{SortByLayer}{V}		
		\State Sort the vertices in $V_u$ by layer, Lamport timestamp and hash in that order.
		\EndFunction
	\end{algorithmic}
\end{algorithm}

With the final ordering computed by using above algorithms, we can assign the consensus time to the finally ordered event blocks.

Then the Atropos is assigned a timestamp from the \texttt{event.MedianTimestamp}, which is computed as a median of the event's timestamps sent by all the nodes.
The median time is used to guarantee aBFT for the consensus time for the Atropos to avoid incorrect time from malicious nodes.
Once Atropos consensus time is assigned, the Clotho becomes an Atropos and each node stores the hash value of Atropos and Atropos consensus time in Main-Chain. The Main-chain is used for time order between event blocks.

\subsection{Block timestamp}

In Lachesis, block timestamp is assigned using Atropos's median time.

\subsection{Peer selection}
In Lachesis protocol, we use a model in which each node only stores a subset of $n$ peers.
To create a new event block, a node first synchronizes the latest blocks from other peers and then it selects the top blocks of at most $k$ peers.
Lachesis protocol does not depend on how peer nodes are selected. 

There are multiple ways to select $k$ peers from the set of $n$ nodes. An simple approach can use random selection from the pool of $n$ nodes. 
In \onlay~\cite{onlay19}, we have tried a few peer selection algorithms as follows:
\begin{itemize}
	\item Random: randomly select a peer from $n$ peers;
	\item Least Used: select the least use peer(s).
	\item Most Used (MFU): select the most use peer(s).
	\item Fair: select a peer that aims for a balanced distribution.
	\item Smart: select a peer based on some other criteria, such as successful throughput rates, number of own events.
\end{itemize}

{Stake-based Peer Selection}: We also propose new peer selection algorithms, which utilize stakes to select the next peer. Each node has a mapping of the peer $i$ and the frequency $f_i$ showing how many times that peer was selected. New selection algorithms define some function $\alpha_i$ that take user stake  $w_i$ and frequency $f_i$ into account. We give a few examples of the new algorithms, described as follows:
\begin{itemize}
	\item Stake-based: select a random peer from $n$ peers with a probability proportional to their stakes $w_i$.
	\item Lowest: select the peer with the lowest value of $\alpha_i$.
	\item Highest: select the peer with the highest value of $\alpha_i$.
	\item Balanced: aim for a balanced distribution of selected peers of a node, based on the values $\alpha_i$.
\end{itemize}

There are other possible criteria can be used a peer selection algorithm, such as successful validation rates, total rewards, etc. A more complex approach is to define a cost model for the node selection, such as low communication cost, low network latency, high bandwidth and high successful transaction throughput.
In Lachesis, available gas and origination power are used.

\subsection{Dynamic participants}
Unlike some existing DAG-based approaches~\cite{hashgraph16}, our Lachesis protocol allows an arbitrary number of participants to dynamically join the system. 

OPERA DAG can still operate with new participants. 
Algorithms for selection of Roots, Clothos and Atroposes are flexible enough and not dependent on a fixed number of participants.

\subsection{Peer synchronization}

Algorithm~\ref{al:syncevents} shows a pseudo code to synchronize events between the nodes.

\begin{algorithm}[htb]
	\caption{EventSync}\label{al:syncevents}
	\begin{algorithmic}[1]
		\Procedure{sync-events()}{}		
		\State Node $n_1$ selects random peer to synchronize with
		\State $n_1$ gets local known events
		\State $n_1$ sends RPC sync request to peer
		\State $n_2$ receives RPC sync request
		\State $n_2$ does an graph-diff check on the known map	
		\State $n_2$ returns unknown events, and mapping of known events to $n_1$		
		\EndProcedure
	\end{algorithmic}
\end{algorithm}

The algorithm assumes that a node always needs the events in topological ordering (specifically in reference to the Lamport timestamps), an alternative would be to use an inverse bloom lookup table (IBLT) for completely potential randomized events.
Alternatively, one can simply use a fixed incrementing index to keep track of the top event for each node.

In Lachesis, when a node receives a coming (new) event block from a peer, it will perform several checks to make sure the event block has been signed by an authorized peer, conforms to data integrity, and has all parent blocks already added in the node's local DAG.

\subsection{Detecting Forks}\label{se:forkdetection}

\dfnn{Fork}{A pair of events ($v_x$, $v_y$) is a fork if $v_x$ and $v_y$ have the same creator, but neither is a self-ancestor of the other. Denoted by $v_x \efork v_y$.}

For example, let $v_z$ be an event in node $n_1$ and two child events $v_x$ and $v_y$ of $v_z$. if $v_x \eself v_z$, $v_y \eself v_z$, $v_x \not \eself v_y$, $v_y \not \eself v_z$, then ($v_x$, $v_y$) is a fork.
The fork relation is symmetric; that is $v_x \efork v_y$ iff $v_y \efork v_x$.

By definition, ($v_x$, $v_y$) is a fork if $cr(v_x)=cr(v_y)$, $v_x \not \eancestor v_y$ and $v_y \not \eancestor v_x$. Using Happened-Before, the second part means $v_x \not \rightarrow v_y$ and $v_y \not \rightarrow v_x$. By definition of concurrent, we get $v_x \concur v_y$.

\begin{lem}
	If there is a fork $v_x \efork  v_y$, then $v_x$ and $v_y$ cannot both be roots on honest nodes.
\end{lem}
Here, we show a proof by contradiction. Any honest node cannot accept a fork so $v_x$ and $v_y$ cannot be roots on the same honest node. Now we prove a more general case. Suppose that both $v_x$ is a root of $n_x$ and $v_y$ is root of $n_y$, where $n_x$ and $n_y$ are honest nodes. Since $v_x$ is a root, it reached events created by nodes having more than $2W/3$. Similarly, $v_y$ is a root, it reached events created by nodes of more than $2W/3$. Thus, there must be an overlap of more than $W$/3. Since we assume less than $W$/3 are from malicious nodes, so there must be at least one honest member in the overlap set. Let $n_m$ be such an honest member. Because $n_m$ is honest, $n_m$ does not allow the fork.

\subsection{Transaction confirmations}

Here are some steps for a transaction to reach finality in our system. 
\begin{itemize}
	\item 
	First, when user submits a transaction into a node, a successful submission receipt will be issued to the client as a confirmation of the submitted transaction. 
	\item Second, the node will batch the submitted transaction(s) into a new event block, and add it into its OPERA DAG. Then will broadcast the event block to all other nodes of the system. Peer nodes will update its own record confirming that the containing event block identifier is being processed.
	\item Third, when the event block is known by majority of the nodes (e.g., it becomes a Root event block), or being known by such a Root block, new status of the event block is updated.
	\item Fourth, our system will determine the condition at which a Root event block becomes a Clotho for being further acknowledged by a majority of the nodes. A confirmation is then sent to the client to indicate that the event block has come to the semi-final stage as a Clotho or being confirmed by a Clotho. After the Clotho stage, we will determine the consensus timestamp for the Clotho and its dependent event blocks. 
	\item Once an event block gets the final consensus timestamp, it is finalized and a final confirmation will be issued to the client that the transaction has been successfully finalized.
\end{itemize}

The above five steps are done automatically by \onlay. Steps (2), (3) and (4) are internal steps. Steps (1) and (5) are visible to the end users. Currently, the whole process will take 1-2 seconds from transaction submission to transaction confirmation (through the five steps). Once a block is confirmed, it will be assigned with a block id in the blockchain and that confirmation is final.

There are some cases that a submitted transaction can fail to reach finality. Examples include a transaction does not pass the validation, e.g., insufficient account balance, or violation of account rules. The other kind of failure is when the integrity of DAG structure and event blocks is not complied due to the existence of compromised or faulty nodes. In such unsuccessful cases, the event block's status is updated accordingly. Clients can always query the latest transaction status regardless of its success or failure.

\subsection{Checkpoint}
\stair\ uses the following procedure, which is similar to the Casper model~\cite{buterin2018}. 
Our approach can be summarized as follows:
\begin{enumerate}
	\item We can divide into checkpoints, each will take every 100th (some configurable number) frame. 
	\item Stakeholders can choose to make more deposits at each check point, if they want to become validators and earn more rewards. Validators can choose to exit, but cannot withdraw their deposits until three months later.
	\item A checkpoint is selected based on a consistent global history that is finalized with more than 2/3 of the validating power for the checkpoint. When a checkpoint is finalized, the transactions will not be reverted. 
	\item With asynchronous system model, validators are incentivized to coordinate with each other on which checkpoints the history should be updated.  Nodes gossip their latest local views to synchronize frames and accounts. Attackers may attempt double voting to target double spending attacks. Honest validators are incentivized to report such behaviors and burn the deposits of the attackers.
\end{enumerate}

Each frame in Lachesis is about 0.3-0.6s, and blocks are confirmed in 0.7-1.5s depending on the number of transactions put into the system. 
To have more deterministic checkpoints, we  divide into epochs, each of which lasts around 4 hours.

\section{Staking model}\label{se:staking}

This section presents our staking model for Opera network. We first give general definitions and variables of our staking model, and then present the model in details.

\subsubsection{Definitions and Variables}

Below are the general definitions and variables that are used in Opera network.

\textbf{General definitions}

\begin{longtable}{p{1cm} p{13cm}}
	$\mathcal{F}$ & denotes the network itself \\[2pt] 
	$SFC$          & "Special Fee Contract" -- managing the collection of transaction fees and the payment of all rewards \\[2pt] 
	$FTM$          & main network token \\[2pt] 
\end{longtable}

\textbf{Accounts}

\begin{longtable}{p{1cm} p{1.2cm} p{13cm}}
	$\mathbb{U}$  & & set of all participant accounts in the network \\[2pt] 
	$\mathbb{A}$  & $\subset \mathbb{U}$ & accounts with a positive FTM token balance \\[2pt] 
	$\mathbb{S}$  & $\subseteq \mathbb{A}$ & accounts that have staked for validation (some of which may not actually be validating nodes) \\[2pt] 
	$\mathbb{V}$  & $\subseteq \mathbb{S}$ & validating accounts, corresponding to the set of the network's validating nodes
\end{longtable}

A participant with an account having a positive FTM token balance, say $i \in \mathbb{A}$, can join the network (as a delegator or a validator).
But an account $i$ in $\mathbb{S}$ may not participate in the protocol yet. Those who join the network belong to the set $\mathbb{V}$.

\dfnn{Validating account}{A validating account has more than $U$ tokens}.

\textbf{Network parameters subject to on-chain governance decisions}

\begin{longtable}{p{1cm} p{1.2cm} p{13cm}}
	$F$  & $3.175e9$ & total supply of FTM tokens \\[2pt] 
	$\lambda$ & $7$ & period in days after which validator staking must be renewed, to ensure activity \\[2pt] 
	$\varepsilon$ & 1 & minimum number of tokens that can be staked by an account for any purpose \\[2pt] 
	$\phi$    & 30\% & SPV commission on transaction fees \\[2pt] 
	$\mu$     & 15\% & validator commission on delegated tokens
\end{longtable}

\textbf{Tokens held and staked}

Unless otherwise specified, any mention of \textit{tokens} refers to FTM tokens.

\dfnn{Token helding}{The token helding  $t_i$ of an account is number of FTM tokens held by account $i \in \mathbb{A}$.}

\begin{longtable}{p{1cm} p{1.2cm} p{13cm}}
	$t_i$        &                 & number of FTM tokens held by account $i \in \mathbb{A}$ \\[2pt] 
	$t_i^{[x]}$      & $> \varepsilon$ & transaction-staked tokens by account $i$\\[2pt] 
	$t_i^{[d]}(s)$   & $> \varepsilon$ & tokens delegated by account $i$ to account $s \in \mathbb{S}$ \\[2pt] 
	$t^{[d]}(s)$     &                 & total of tokens delegated to account $s \in \mathbb{S}$ \\[2pt] 
	$t_i^{[d]}$      &                 & total of tokens delegated by account $i$ to accounts in $\mathbb{S}$ \\[2pt] 
	$t_i^{[s]}$      &                 & validation-staked tokens by account $i$ \\[2pt] 
\end{longtable}

The sum of tokens staked or delegated by an account $i \in \mathbb{A}$ cannot exceed the amount of tokens held:
$
t_i^{[x]} + t_i^{[s]} + \sum_{s \in \mathbb{S}} t_i^{[d]}(s)  \leq t_i
$.
The total amount of tokens delegated to an account $s \in \mathbb{S}$ is: $
t^{[d]}(s) = \sum_{i \in \mathbb{A}} t_i^{[d]}(s)
$. 
The total amount of tokens delegated by an account $i \in \mathbb{A}$ is: $
t_i^{[d]} = \sum_{s \in \mathbb{S}} t_i^{[d]}(s)
$.

\subsubsection{Validating power}

We use a simple model of validating power, which is defined as the number of tokens held by an account.
The weight of an account $i \in \mathbb{A}$ is equal to its token holding $t_i$.

\subsubsection{Block consensus and Rewards}
\dfnn{Validation score} {Validation score of a block is the total validating power that a given block can achieve from the validators $v \in \mathbb{V}$.}

\dfnn{Quorum}{Validating threshold is defined by $2/3$ of the validating power that is needed to confirm an event block to reach consensus.}

Below are the variables that define the block rewards and their contributions for participants in Fantom network.

\begin{longtable}{p{1.2cm} p{2.5cm} p{11cm}}
	$Z$      & 996,341,176 & total available block rewards of $FTM$, for distribution by the $SPV$ during the first 1460 days after mainnet launch\\[2pt] 
	$F_s$    &   & $FTM$ tokens held by the $SPV$ \\[2pt] 
	$F_c$    & $F - F_s$ & total circulating supply \\[2pt] 	
\end{longtable}

Block rewards will be distributed over each epoch.

\subsubsection{Token Staking and Delegation}\label{se:stakingoverview}
Participants can chose to stake or delegate their FTM tokens. When staking or delegating, the validating power of a node is based on the number of FTM tokens held.
Three possible ways to stake are given as follows.

\dfnn{Transaction staking}{Participants can gain more stakes or tokens via the transaction staking.}
Transaction submitter can gain reward from the transaction fees of submitted transactions.
This style of staking helps increases transaction volume on the network. The more liquidity, the more transaction staking they can gain.

\dfnn{Validation staking}{By validating blocks, a participant can gain validation rewards. Honest participants can gain block rewards for  successfully validated blocks.}

Participants can achieve block rewards for blocks that they co-validated and gain transaction fees from transaction submitter for the successfully finalized transactions. The more stake they gain, the more validating power they will have and thus the more rewards they can receive as it is proportional to their validating power.

\dfnn{Validation delegation}{Validation delegation allows a participant to delegate all or part of their tokens to another participant(s). Delegating participants can gain a share of block rewards and transaction fees, based on the amount of delegated stake.}

Delegators can delegate their stake into a validator or into multiple validators. Participants with large amount of stake can delegate their stakes into multiple validators. Delegators earn rewards from their delegations.
Validators will not be able to spend delegated tokens, which will remain secured in the stakeholder’s own address. Validators will receive a fixed proportion of the validator fees attributable to delegators.

\dfnn{Validation performance}
Participants are rewarded for their staked or delegated amount. Delegators will be incentivized to choose validator nodes that have a high self-stake, i.e. are honest and high performing.
Delegators can delegate their tokens for a maximum period of days, after which they need to re-delegate. The requirements to delegate are minimal:
\begin{itemize}
	\item \textit{Minimum number of tokens per delegation}: 1
	\item \textit{Minimum lock period}: 1 day
	\item \textit{Maximum lock period}: 1 year
	\item \textit{Maximum number of delegations made by a user}: None
	\item \textit{Maximum number of delegated tokens  per validator}: 15 times of the validator's tokens.
\end{itemize}

\section{Opera Network}\label{se:network}

\subsection{Network Architecture}
Figure~\ref{fig:stairframework} depicts an overview of Opera network, as introduced in our \stair\ framework~\cite{stairdag}. 

A Validator node consists of three components: state machine, consensus and networking.  A application can communicate to a node via CLI. Opera network supports auditing by permitting participants to join  in post-validation mode. An observer (or Monitor) node consists of state machine, post validation component and networking component.

\begin{figure}[ht]
	\centering
	\includegraphics[width=0.55\linewidth]{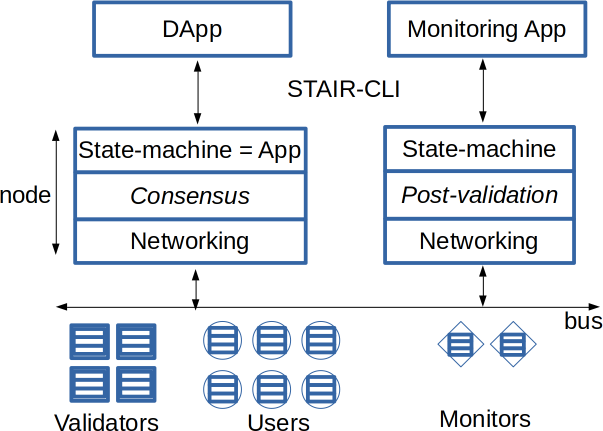}
	\caption{An overview of Opera network}
	\label{fig:stairframework}
\end{figure}

\subsection{Nodes and Roles}
For an asynchronous DAG-based system, Lachesis supports three types of participants: users,  validators and monitors.

Each validating node can create new event blocks. Generation of a new even block indicates that the new block and all of its ancestors have been validated by the creator node of that new event block.

Users have less than $U$ FTM tokens, where $U$ is a predefined threshold value. The current threshold value, which is voted by the community is $U$ = 1,000,0000 FTMs. Each validator must have more than $U$ tokens. In Lachesis, users can only delegate their stake to validator nodes. User nodes cannot run as a validating node, which can create and validate event blocks.

Besides validating nodes, Opera network allows another type of participants --- observers, who can join the network, but do not act as validating nodes. Observers are not required to have any stake to join the network and thus they cannot generate nor perform online voting (online validation). But observers can do post-validation. This is to encourage the community to join the checks-and-balances of our public network.

\subsection{Boot Node Service}
A new node that joins Opera network for the first time will connect to one of our boot nodes, which provides a discovery service to connect with other peers in the network. Currently, there are five boot nodes in Opera network and the boot nodes are distributed in different zones. Boot nodes serve RPC requests from peers and help new nodes to discover peers in the network.

\subsection{API Service}
Since the first launch of Opera mainnet, the network has served lots of dApp projects and users. We have scaled our infrastructure on EC2 AWS accordingly to cope with a huge amount of requests.

There are currently about 30 servers used for RPC API. The RPC API load is about 250,000 requests per second at peaks. Each server can serve about 9000 requests per second. They are run in 3 hives that are independently serviced by their own balancing gateways and the gateways are switched by round robin DNS setup in 60s roundtrip time. 
Also, we have additional 3 public RPC endpoints for API calls that are use for backup service.

Our infrastructure include other servers to serve explorer and wallet service. We also maintain a different set of servers for cross-chain integration.

Apart from the servers run by the Foundation, there are many more API servers operated by our community members, who are active contributors to the Opera network.

\subsection{Mainnet}
There are currently 46 validator nodes, 5 boot nodes and more than 50 servers running API service.
Our Mainnet has served more than 35 million transactions, and have produced more than 12.8 million blocks. There are more than 300k accounts.

\subsection{Testnet}
There are currently 7 validator nodes and 2 API nodes. The testnet has been used by developers to test and deploy new smart contracts and to prepare for new integration prior into launch it on our mainnet.

\subsection{Implementation}

We have implemented our Lachesis protocol in GoLang 1.15+.
Our implementation is available at \url{https://github.com/Fantom-foundation/go-opera}. Our previous repository is available at \url{https://github.com/Fantom-foundation/go-lachesis}.

Opera network is fully EVM compatible and smart contracts can be deployed and run on our Opera network.
We implemented Lachesis on top of Geth~\footnote{https://github.com/ethereum/go-ethereum}.  We provide Web3 API calls support through our API Service layer. 
Decentralized applications (dApps) can query and send transactions into the network for execution. Documentation and guides are available at \url{https://docs.fantom.foundation/}.

Figure~\ref{fig:lachesis-architecture} shows the overall architecture of a Lachesis node.
\begin{figure}
	\centering
	\includegraphics[width=0.8\linewidth]{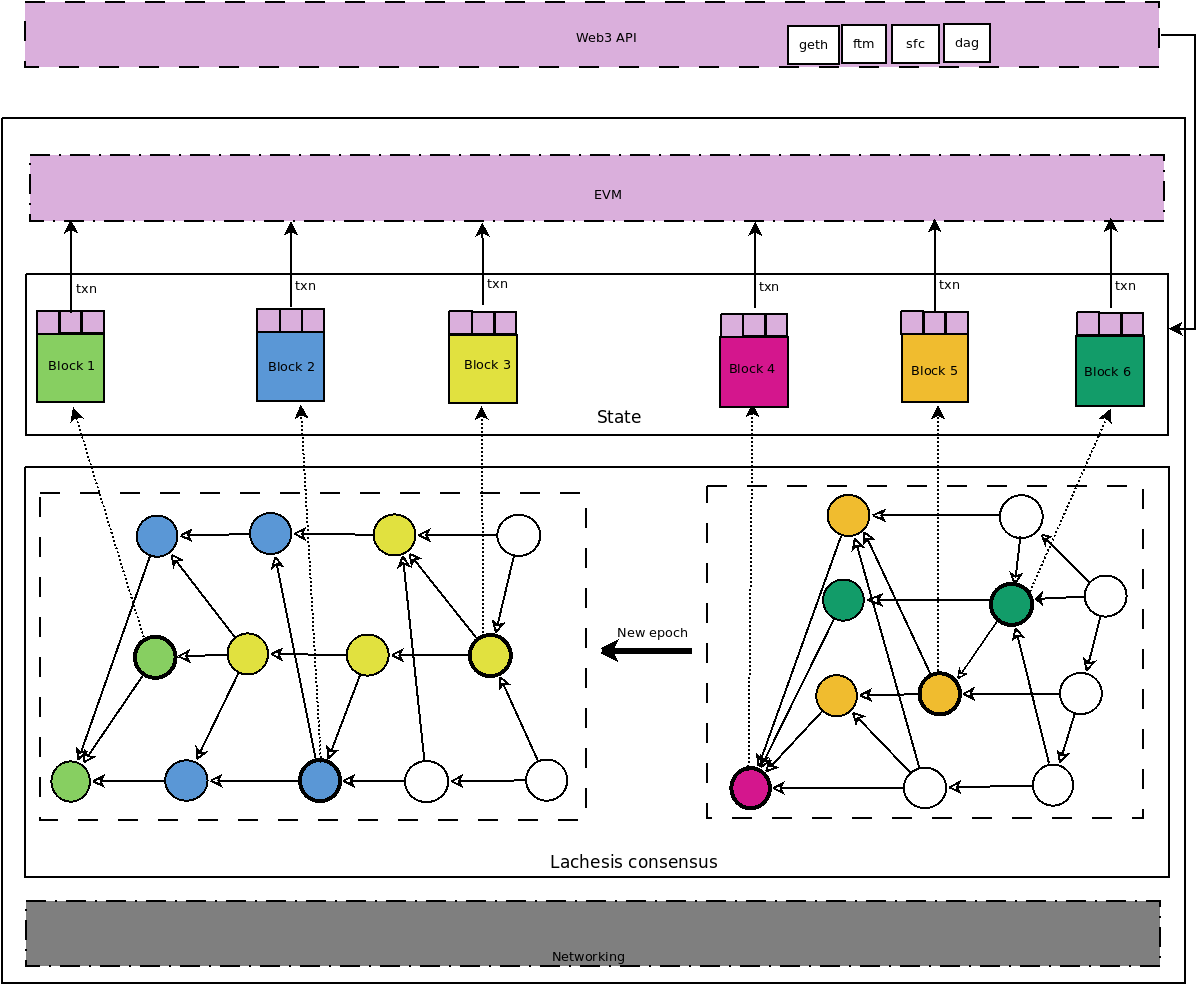}
	\caption{Architecture of a Lachesis node}
	\label{fig:lachesis-architecture}
\end{figure}

{\bf Transaction Fee}:
Like many existing platforms, Opera network leverages transaction fee, which is paramount to prevent transaction flooding. It will cost attackers when they attempt to send a lot of spamming transactions to the network.

Clients can send transactions to Opera network. Each transaction requires a small amount of fee to compensate for the cost to execute the transaction on Validator nodes. Transaction can range from simple account transfer, new smart contract deployment or a complex smart contract call. Transaction fee is proportional to the mount of op codes included in the transaction.

In Ethereum and few others, transaction fees are calculated in a separate gas token rather than the native token of the network. Instead, Fantom's Opera network calculates transaction fees in native FTM. This design removes the burden of acquiring two types of token for users and developers to be able to use and run dApps on our Opera mainnet.

{\bf Transaction Pool}: 
Transactions once submitted will be appended into a transaction pool of a node. 
The transaction pool is a buffer to receive incoming transactions sent from peers.
Similar to Ethereum's txpool implementation, transactions are placed in a priority queue in which transactions with higher gas will have more priority. 

Each node then pulls transactions from the pool and batches them into new event block.
The number of new event blocks a node can create is proportional to the validating power it has.

\subsection{Performance and Statistics}

For experiments and benchmarking, we set up a private testnet with 7 nodes on EC2 AWS. Each node is running m2.medium instance, which is 2vCPU with 3.3 GHz each, 4GB RAM and 200GB SSD storage.
We also experimented using a local machine with the following specs (CPU: AMD Ryzen 7 2700, Memory: 2x DDR4/16G 2.6 Ghz, SSD storage).

\subsubsection{Normal emission rate and gas allocation}

We ran experiments with the private testnet with 7 nodes. The experiment was run using normal emission rate (200ms) and normal gas allocation rules. We ran the local machine to sync with the testnet.

On the local machine, the maximum raw TPS is 11000, and maximum syncing TPS is ~4000-10000 TPS.
The total number of transactions executed is 14230, the total of events emitted is 109, and there are 108 finalized blocks.
The peak TPS for P2P syncing on a local non-validator node (16 CPUs, SSD) is 3705.

Our results on a validator node show that the peak TPS for P2P syncing is 2760.
Maximum syncing TPS is ~3000-7000 on testnet server, whereas maximum network TPS is 500 due to txpool bottleneck on testnet servers.

\subsubsection{High latency nodes}
We also ran an experiment where 1/3 of the nodes in the network were lagged. Additional latency of 300-700ms was added into 3 validators, whereas normal latency was used for others.

\begin{table}[ht]
\centering
\begin{tabular}{ |l |c |c|}
\hline
Test duration (seconds)	& 485 & 793 \\
Blocks finalized &	1239 &	605 \\
Transactions executed &	24283 &	39606 \\
Events emitted	& 10930 &	5959 \\
Average TTF (seconds) & 0.92 &	4.64 \\
Average TPS	& 50.06 &	49.94 \\
\hline
\end{tabular}
\caption{Statistics of with 3 nodes of high latency}
\label{fig:highlatency-ex}
\end{table}
Table~\ref{fig:highlatency-ex} shows that the added latency into 3 nodes can increase TTF by 4-fold, but TPS is slightly reduced.

\subsubsection{Emission rates}

We ran experiments in that each node will generate 1428 transactions per second and thus the whole network of 7 test servers generated 10000 transactions per second. Each generated transaction has around 100 bytes, like a simple FTM transfer.
To benchmark the consensus, this experiment disabled transaction execution, so it only counted the transaction creation, consensus messages, and finalized transactions and blocks.

\begin{figure}[ht]
	\centering
	\subfloat[tps]{\includegraphics[width=0.45\linewidth]{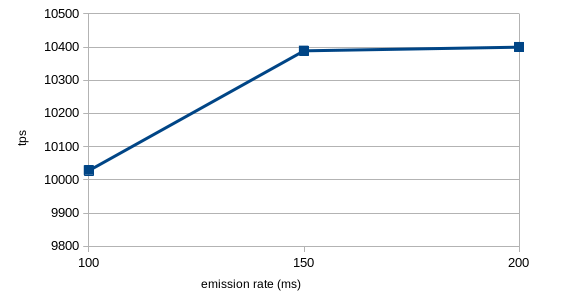}}
	\qquad
	\subfloat[ttf]{\includegraphics[width=0.45\linewidth]{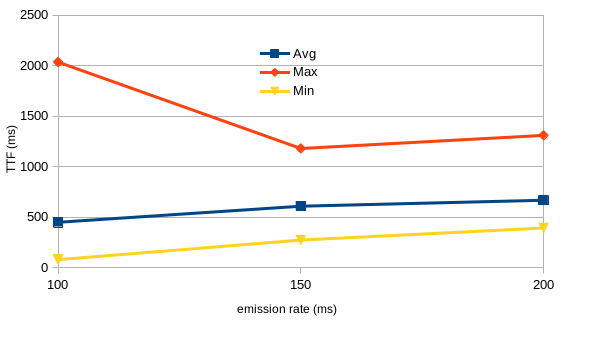}}
	\caption{TPS and TTF measures on a testnet with 7 nodes}
	\label{fig:7node-exp}
\end{figure}

Figure~\ref{fig:7node-exp} shows the statistics of block TPS and block TTF when 10000 transactions are submitted into the network per second. We experimented with different emission rates (100ms, 150ms and 200ms), which is the amount of time each new event created per node. As depicted in the figure, when increase emission rate from 100ms to 200ms, the tps increased and the average TTF also increased.

\section{Discussions}\label{se:discuss}

\subsection{Comparison with Existing Approaches}

Compared with existing PoS protocols ~\cite{pass2017fruitchains,king2012ppcoin,vasin2014blackcoin,bentov2016,kwon2014tendermint,algorand16}, our Lachesis protocol is built on top of DAG and it is more deterministic and more secure.

Compared with existing DAG-based approaches ~\cite{dagcoin15, sompolinsky2016spectre, byteball16, hashgraph16, tangle17, PHANTOM08, PARSEC18, conflux18}, Lachesis algorithms are quite different. Among them, Hashgraph~\cite{hashgraph16} is the closest to ours. 
However, our Lachesis algorithm is quite different from Hashgraph's algorithm in many ways:
\begin{itemize}
	\item Lachesis algorithm for finding roots and Atropos is deterministic, whereas Hashgraph uses coin round to decide 'famous' event blocks. 
	\item Our Lachesis leverages the notion of Happened-before relation, Lamport timestamp, and graph layering to reach consensus on the final ordering for Atropos blocks. 
	\item Our algorithm to create blocks (i.e., finalized blocks of the blockchains) is unique. New block is created for every newly decided Atropos, whereas Hashgraph tends to process with multiple Famous event blocks at once.
	\item The event blocks under Atropos are sorted based on Lamport timestamp and graph layering information. On the contrary, Hashgraph uses a heuristics, so-called \emph{findOrder} function described in their paper, to compute the ordering of events as well as their timestamps. 
	\item Lachesis protocol is permissionless aBFT, whereas it is unclear if their algorithm is truly BFT.
	\item Fantom's Opera network is a permissionless (public) network and our Lachesis protocol has run successfully. In contrast, Hashgraph aims for a permissioned network.
	\item We give a detailed model of our Proof-of-Stake model in StakeDag and StairDag framework~\cite{stakedag, stairdag}, whereas PoS is briefly mentioned in their Hashgraph paper~\cite{hashgraph16}.
	\item We also provide a formal semantics and proof for our Lachesis protocol using common concurrent knowledge~\cite{cck92}.
\end{itemize}

\subsection{Protocol Fairness and Security}

There has been extensive research in the protocol fairness and security aspects of existing PoW and PoS blockchains (see surveys~\cite{sheikh2018proof, panarello2018survey}). 
Here, we highlight the benefits of our proposed PoS+DAG approach.

\subsubsection{Protocol Fairness}
Regarding the fairness, PoW protocol is fair because a miner with $p_i$ fraction of the total computational power can win the reward and create a block with the probability $p_i$.
PoS protocol is fair given that an individual node, who has $w_i$ fraction of the total stake or coins, can a new block with $w_i$ probability.
However, in PoS systems, initial holders of coins tend to keep their coins in their balance in order to gain more rewards.

Our Lachesis protocol is fair since every node has an equal chance to create an event block. Nodes in Lachesis protocol can enter the network without an expensive hardware like what is required in PoW. Further, any node in Lachesis protocol can create a new event block with a stake-based probability, like the block creation in other PoS blockchains.

Like a PoS blockchain system, it is a possible concern that the initial holders of coins will not have an incentive to release their coins to third parties, as the coin balance directly contributes to their wealth. 
Unlike existing PoS protocols, each node in Lachesis protocol are required to validate parent event blocks before ot can create a new block. Thus, the economic rewards a node earns through event creation is, in fact, to compensate for their contribution to the onchain validation of past event blocks and it's new event block.

\begin{table}[h]
	\centering
	\begin{tabular}{|l|c|c|c|}
		\hline
		Criterion & PoW & PoS & Lachesis \\\hline
		Block creation probability & $p_i /P$ & $w_i /W $ & $w_i/W$ \\\hline
		Validation probability & $1/n$ & $w_i/W$ &  $1/n$ \\\hline
		Validation reward & $p_i /P$ & $w_i /W$ & $w_i /W$ \\\hline
		Txn Reward & $p_i /P$ & $w_i /W$ & $w_i /W$ \\\hline
		Performance reward & - & - & (+$1$ Saga point) \\
		\hline
	\end{tabular}
	\vspace{5pt}
	\caption{Comparison of PoW, PoS and Lachesis protocol}\label{tab:comp}
\end{table}

Table~\ref{tab:comp} gives a comparison of PoW, PoS and Lachesis protocols.
Let $p_i$ denote the computation power of a node and $P$ denote the total computation power of the network. Let $w_i$ be the stake of a node and $W$ denote the total stake of the whole network. Let $\alpha_i$  denote the number of Saga points a node has been rewarded from successful validations of finalized blocks.

Remarkably, our Lachesis protocol is more intuitive because our reward model used in stake-based validation can lead to a more reliable and sustainable network.

\subsubsection{Security}
Regarding security, Lachesis has less vulnerabilities than PoW, PoS and DPoS.
As for a guideline, Table~\ref{tab:vulnerability} shows a comparison between existing PoW, PoS, DPoS and our Lachesis protocol, with respect to the effects of the common types of attack on previous protocols. Note that, existing PoS and DPos approaches have addressed one or some of the known vulnerabilities in one way or the other. It is beyond the scope of this paper to give details of such approaches.

\begin{table}[ht]
	\centering
	\begin{tabular}{|l|c|c|c|c|}
		\hline
		Attack type & PoW & PoS & DPoS & Lachesis (PoS+DAG)\\
		\hline
		Selfish mining & ++ & - & - & - \\
		Denial of service & ++ & + & + & + \\
		Sybil attack & ++ & + & + & - \\
		Precomputing attack & - & + & - & - \\
		Short range attack (e.g., bribe) & - & + & - & - \\
		Long range attack  & - & + & + & - \\
		Coin age accumulation attack & - & maybe & - & - \\
		Bribe attack & + & ++ & + & - \\
		Grinding attack & - & + & + & - \\
		Nothing at stake attack & - & + & + & - \\
		Sabotage & - & + & + & + \\
		Double-Spending & - & + & + & - \\
		\hline
	\end{tabular}
	\vspace{5pt}
	\caption{Comparison of Vulnerability Between PoW, PoS, DPoS and Lachesis Protols}\label{tab:vulnerability}
\end{table}

PoW-based systems are facing \emph{selfish mining attack}, in which an attacker selectively reveals mined blocks in an attempt to waste computational resources of honest miners.

Both PoW and PoS share some common vulnerabilities. \emph{DoS attack} and \emph{Sybil attack} are shared vulnerabilities, but PoW are found more vulnerable. A DoS attack disrupts the network by flooding the nodes. In a Sybil attack, the attacker creates multiple fake nodes to disrupt the network. Another shared vulnerable is \emph{Bribe attack}. In bribing, the attacker performs a spending transaction, and at the same time builds an alternative chain secretly, based on the block prior to the one containing the transaction. After the transaction gains the necessary number of confirmations, the attacker publishes his/her chain as the new valid blockchain, and the transaction is reversed. PoS is more vulnerable because a PoS Bribe attack costs 50x lower than PoW Bribe attack.

PoS has encountered issues that were not present in PoW-based blockchains. These issues are: (1) \emph{Grinding attack:} malicious nodes can play their bias in the election process to gain more rewards or to double spend their money: (2) \emph{Nothing at stake attack:} A malicious node can mine on an alternative chain in PoS at no cost, whereas it would lose CPU time if working on an alternative chain in PoW.

There are possible attacks that are only encountered in PoS. Both types of attack can induce conflicting finalized checkpoints that will require offline coordination by honest users.

\emph{Double-Spending}
An attacker (a) acquires $2W/3$ of stakes; (b) submits a transaction to spend some amount and then votes to finalize that includes the transactions; (c) sends another transaction to double-spends; (d) gets caught and his stakes are burned as honest validators have incentives to report such a misbehavior. 
In another scenario, an attacker acquires (a) $W/3+\epsilon$ to attempt for an attack and suppose Blue validators own $W/3-\epsilon/2$, and Red validators own the rest $W/3-\epsilon/2$. The attacker can (b) vote for the transaction with Blue validators, and then(c) vote for the conflicting transaction with the Red validations. Then both transactions will be finalized because they have $2W/3 +\epsilon/2$ votes. Blue and Red validators may later see the finalized checkpoint, approve the transaction, but only one of them will get paid eventually.

\emph{Sabotage (going offline)}  An attacker owning $W/3 +\epsilon$ of the stakes can appear offline by not voting and hence checkpoints and transactions cannot be finalized. Users are expected to coordinate outside of the network to censor the malicious validators.

{\bf Attack cost in PoS versus PoW} It will cost more for an attack in PoS blockchain due to the scarcity of the native token than in a PoW blockchain.
In order to gain more stake for an attack in PoS, it will cost a lot for an outside attacker. 
S/he will need $2W/3$ (or $W/3$ for certain attacks) tokens, where $W$ is the total number of tokens, regardless of the token price. Acquiring more tokens will definitely increase its price, leading to a massive cost. Another challenge is that all the tokens of a detected attempt will be burned.
In contrast, PoW has no mechanism nor enforcement to prevent an attacker from reattempting another attack. An attacker may purchase or rent the hash power again for his next attempts.

Like PoS, our Lachesis protocol can effectively prevent potential attacks as attackers will need to acquire $2W/3$ tokens (or at least $W/3$ for certain attacks) to influence the validation process. Any attempt that is detected by peers will void the attacker's deposit.

\subsection{Response to Attacks}\label{se:ra}

Like other decentralized blockchain technologies, Fantom platform may also face potential attacks by attackers.
Here, we present several possible attack scenarios that are commonly studied in previous work, and we show how Lachesis protocol can be used to prevent such attacks.

{\bf Transaction flooding}:
A malicious participant may run a large number of valid transactions from their accounts with the purpose of overloading the network. In order to prevent such attempts, our Fantom platform has applied a minimal transaction fee, which is still reasonable for normal users, but will cost the attackers a significant amount if they send a large number of transactions. 

{\bf Parasite chain attack}:
In some blockchains, malicious node can make a parasite chain with an attempt to make a malicious event block. In Lachesis protocol, when a consensus is reached, finalized event block is verified before it is added into the Main chain. 

Our Lachesis protocol is 1/3-BFT, which requires less than one-third of nodes are malicious. The malicious nodes may create a parasite chain. As root event blocks are known by nodes with more than 2W/3 validating power, a parasite chain can only be shared between malicious nodes, which are accounted for less than one-third of participating nodes. Nodes with a parasite chain are unable to generate roots and they will be detected by the Atropos blocks.

{\bf Double spending}:
A double spending attack is when a malicious actor attempts to spend their funds twice. 

Let us consider an example of double spending. Entity $A$ has 10 tokens, but it sends 10 tokens to $B$ via node $n_A$ and at the same time it also sends 10 tokens to $C$ via node $n_C$. Both node $n_A$ and node $n_C$ agree that the transaction is valid, since $A$ has the funds to send to $B$ (according to $n_A$) and $C$ (according to $n_C$).
In case one or both of $n_A$ and $n_C$ are malicious, they can create a fork $x$ and $y$.

Lachesis protocol can structurally detect the fork by all honest nodes at some Atropos block. 
Specially, let us consider the case an attacker in Lachesis protocol with $W/3 + \epsilon$. Assume the attacker can manage to achieve a root block $r_b$ with Blue validators and a conflict root block $r_r$ with Red validators. However, in order for either of the two event blocks to be finalized (becoming Clotho and then Atropos), each of root blocks need to be confirmed by roots of next levels. Since peers always propagate new event blocks, an attacker cannot stop the Blue and Red validators (honest) to sent and receive event blocks to each other. Therefore, there will exist an honest validator from Blue or Red groups, who detects conflicting event blocks in the next level roots aka Atroposes. Because honest validators are incentivized to detect and prevent these misbehavior, the attacker will get caught and his stakes are burned.

{\bf Sabotage attack}:
Similarly, for Sabotage attack, the attacker may refuse to vote for a while. 
In Lachesis, honest validators will find out the absence of those high stake validators after some time period.
Validators what are offline for a while will be pruned from the network. After that, those pruned validators will not be counted and the network will function like normal with the rest of nodes.

{\bf Long-range attack}:
In some blockchains such as Bitcoin~\cite{bitcoin08}, an adversary can create another chain. If this chain is longer than the original, the network will accept the longer chain because the longer chain has had more work (or stake) involved in its creation.

In Lachesis, this long-range attach is not possible. Adding a new block into the block chain requires an agreement of $2n/3$ of nodes (or $2W/3$ in our PoS model). To accomplish a long-range attack, attackers would first need to create more than $2n/3$ malicious nodes (or gain validating power of $2W/3$ in our PoS model) to create the new chain.

{\bf Bribery attack}:
An adversary could bribe nodes to validate conflicting transactions. Since $2n/3$ participating nodes (or nodes with a total more than $2W/3$) are required, this would require the adversary to bribe more than $n/3$ of all nodes (or $W/3$) to begin a bribery attack.

{\bf Denial of Service}:
Lachesis protocol is leaderless requiring $2n/3$ participation. An adversary would have to attack more than $n/3$ (or $W/3$) validators to be able to successfully mount a DDoS attack.

{\bf Sybil}:
Each participating node must stake a minimum amount of FTM tokens to participate in the network. But staking $2/3$ of the total stake would be prohibitively expensive.

\subsection{Choices of Stake Decentralization}

This section presents discussions and benefits of our proposed DAG-based PoS approach.

\subsubsection{Low staked validators}
To have more validating nodes, \stair\ framework introduces a model which allows Users with small of tokens can still run a validating node. Details are given in the paper~\cite{stairdag}. Here, we give a summary. 

There are three general types of nodes that we are considered in our new protocol. They are described as follows.
\begin{itemize}
	\item{\bf{Validators}} {Validator node can create and validate event blocks. They can submit new transactions via their new event blocks. To be a validator, a node needs to have more than $U$ = 1,000,000 FTM tokens. The validating power of a Validator node is computed by $w_i = U \times \lfloor \frac{t_i}{U} \rfloor$.
	}
	\item{\bf{Users}} {User node can create and validate event blocks. They can submit new transactions via their new event blocks. To be a user, a node needs to have more than 1 FTM token. The number of tokens of a User is from 1 to $U$-1 = 999,999 FTM tokens. All users have the same validating power of $w_i$ = 1.
	}	
	\item{\bf{Observers:} } {Observer node can retrieve even blocks and do post-validation. Observers cannot perform onchain voting (validation). Observers have zero FTM tokens and thus zero validating power (e.g., $w_i$ = 0).
	}
\end{itemize}

The value of $U$ is system configurable, which is determined and updated by the community. Initially, $U$ can be set to a high value, but can be reduced later on as the token value will increase. Note that, there is a simple model in which validators have the validating power $w_i$ = $U$, instead of $U \times \lfloor \frac{t_i}{U} \rfloor$.

\subsubsection{Choices of $L$ and $U$}

The value of lower bound $L$=$1$ is used to separate between Observers and Users. The upper bound $U$=$1,000,000$ separates between Users and Validators. Both $L$ and $U$ are configurable and will vary based on the community's decision.

Figure~\ref{fig:stairaccounts} shows the differences between three different account types. The level of trust is proportional to the number of tokens. Observers have a trust level of zero and they have incentives to do post-validation for Saga points.
Users have low trust as their token holding varies from 1 to $U$, whereas Validators have high trust level with more than $U$ tokens in their account. Both Users and Validators have incentives to join the network to earn transaction fees and validation rewards.
\begin{figure}[ht]
	\centering
	\includegraphics[width=0.8\linewidth]{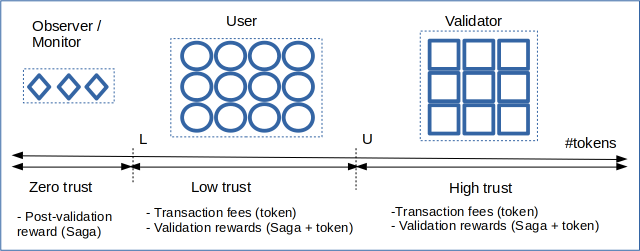}
	\caption{Token holding, trust levels and incentives of the three account types}
	\label{fig:stairaccounts}
\end{figure}

There are a limited number of Observers and a small pool of Validators in the entire network. The number of Users may take a great part of the network.

\subsubsection{Alternative models for gossip and staking}
We also propose an alternative model to ensure a more balanced chance for both Users and Validators. This is to improve liveness and safety of the network. Specifically, each User node randomly selects a Validator node in $\mathcal{V}$ and each Validator node randomly chooses a User node in $\mathcal{U}$. The probability of being selected is proportional to the value $w_i$ in the set, in both cases.
Alternatively, each Validator node can randomly select a node  from all nodes with a probability reversely proportional to the stake $w_i$.

\subsubsection{Network load}
More observers may consume more network capacity from the validating nodes. To support more nodes, observers may not be allowed to synchronize and retrieve events directly from Validators. Instead, we suggest to use a small pool of Moderator nodes that will retrieve updated events from Validators. Observers will then synchronize with the Moderator nodes to download new updates and to perform post-validation. Similarly, low staked validators can synchronize with such Moderator nodes as well.

\section{Conclusion}\label{se:con}

In this paper, we consolidate the key technologies and algorithms of our Lachesis consensus protocol used in Fantom's Opera network. Lachesis protocol integrates Proof of Stake into a DAG-based model to achieve more scalable and fast consensus. It guarantees asynchronous practical BFT. Further, we present a staking model.

By using Lamport timestamps, Happened-before relation, and layer information, our algorithm for topological ordering of event blocks is deterministic and guarantees a consistency and finality of events in all honest nodes. 

We present a formal proofs and semantics of our PoS+DAG Lachesis protocol in the Appendix. Our work extends the formal foundation established in our previous papers~\cite{fantom18, onlay19, stakedag}, which is the first that studies concurrent common knowledge sematics~\cite{cck92} in DAG-based protocols. Formal proofs for our layering-based protocol is also presented.

\subsection{Future work}
We are considering several models to improve the scalability and sustainability of Opera network. In particular, we investigate several approaches to allow more validating nodes to join the network, whilst we aim to improve decentralization, increase TPS and keep TTF low. A brief introduction of these approaches is given in the Appendix.

\newpage
\section{Appendix}

\subsection{Basic Definitions}

Lachesis protocol is run via nodes representing users' machines which together create the Opera network. The basic units of the protocol are called event blocks - a data structure created by a single node, which contains transactions. These event blocks reference previous event blocks that are known to the node. This flow or stream of information creates a sequence of history.

\dfnn{Lachesis}{Lachesis protocol is a consensus protocol}

\dfnn{Node}{Each machine that participates in the Lachesis protocol is called a \emph{node}. Let $n$ denote the total number of nodes.}

\dfnn{$k$}{A \emph{constant} defined in the system.}

\dfnn{Peer node}{A node $n_i$ has $k$ \emph{peer nodes}.}

\dfnn{Process}{A process $p_i$ represents a machine or a \emph{node}. The process identifier of $p_i$ is $i$. A set $P$ = \{1,...,$n$\} denotes the set of process identifiers.}

\dfnn{Channel}{A process $i$ can send messages to process $j$ if there is a channel ($i$,$j$). Let $C$ $\subseteq$ \{($i$,$j$) s.t. $i,j \in P$\} denote the set of channels.}

\subsubsection{Events}

\dfnn{Event block}{Each node can create event blocks, send (receive) messages to (from) other nodes.}

The structure of an event block includes the signature, generation time, transaction history, and hash information to references. All nodes can create event blocks. The first event block of each node is called a \emph{leaf event}.

Suppose a node $n_i$ creates an event $v_c$ after an event $v_s$ in $n_i$.  Each event block has exactly $k$ references. One of the references is self-reference, and the other $k$-1 references point to the top events of $n_i$'s $k$-1 peer nodes.

\dfnn{Top Event}{An event $v$ is a top event of a node $n_i$ if there is no other event in $n_i$ referencing $v$.}

\dfnn{Height Vector}{The height vector is the number of event blocks \emph{created} by the $i$-th node.}

\dfnn{Ref}{An event $v_r$ is called ``ref" of event $v_c$ if the reference hash of $v_c$ points to the event $v_r$. Denoted by $v_c \eref v_r$. For simplicity, we can use $\erefz$ to denote a reference relationship (either $\eref$ or $\eself$).}

\dfnn{Self-ref}{An event $v_s$ is called ``self-ref" of event $v_c$, if the self-ref hash of $v_c$ points to the event $v_s$. Denoted by $v_c \eself v_s$.}

\dfnn{Event references}{Each event block has at least $k$ references. One of the references is self-reference, and the others point to the top events of other peer nodes.}

\dfnn{Self-ancestor}{An event block $v_a$ is self-ancestor of an event block $v_c$ if there is a sequence of events such that $v_c \eself v_1 \eself \dots \eself v_m \eself v_a $. Denoted by $v_c \eselfancestor v_a$.}

\dfnn{Ancestor}{An event block $v_a$ is an ancestor of an event block $v_c$ if there is a sequence of events such that $v_c \erefz v_1 \erefz \dots \erefz v_m \erefz v_a $. Denoted by $v_c \eancestor v_a$.}

For simplicity, we simply use $v_c \eancestor v_s$ to refer both ancestor and self-ancestor relationship, unless we need to distinguish the two cases.

\subsubsection{OPERA DAG}

In Lachesis protocol, the history of event blocks form a directed acyclic graph, called OPERA DAG. OPERA DAG $G=(V,E)$ consists of a set of vertices $V$  and a set of edges $E$ . 
A path in $G$ is a sequence  of vertices ($v_1$, $v_2$, $\dots$, $v_k$) by following the edges in $E$ such that it uses no edge more than once.
Let $v_c$ be a vertex in $G$.
A vertex $v_p$ is the \emph{parent} of $v_c$ if there is an edge from $v_p$ to $v_c$.
A vertex $v_a$ is an \emph{ancestor} of $v_c$ if there is a path from $v_a$ to $v_c$.

\begin{defn}[OPERA DAG]
	OPERA DAG is a DAG graph $G = (V, E)$ consisting of $V$ vertices and $E$ edges. Each vertex $v_i \in V$ is an event block. An edge $(v_i,v_j) \in E$ refers to a hashing reference from $v_i$ to $v_j$; that is, $v_i \erefz v_j$.
\end{defn}

\begin{defn}[vertex]
	An event block is a vertex of the OPERA DAG.
\end{defn}

Suppose a node $n_i$ creates an event $v_c$ after an event $v_s$ in $n_i$. Each event block has at most $k$ references. Thus, each vertex has at most $k$ out-edges. One of the references is self-reference, and the other $k$-1 references point to the top events of $n_i$'s $k$-1 peer nodes.

\subsubsection{Happened-Before relation}
The ``happened before" relation, denoted by $\rightarrow$, gives a partial ordering of events from a distributed system of nodes. 
Each node $n_i$ (also called a process) is identified by its process identifier $i$. 
For a pair of event blocks $v$ and $v'$, the relation "$\rightarrow$" satisfies: (1) If $v$ and $v'$ are events of process $P_i$, and $v$ comes before $v'$, then $b \rightarrow v'$. (2) If $v$ is the send($m$) by one process and $v'$ is the receive($m$) by another process, then $v \rightarrow v'$. (3) If $v \rightarrow v'$ and $v' \rightarrow v''$ then $v \rightarrow v''$. 
Two distinct events $v$ and $v'$ are said to be concurrent if $v \nrightarrow v'$ and $v' \nrightarrow v$.

\dfnn{Happened-Immediate-Before}{An event block $v_x$ is said Happened-Immediate-Before an event block $v_y$ if $v_x$ is a (self-) ref of $v_y$. Denoted by $v_x \hibefore v_y$.}

\dfnn{Happened-Before}{An event block $v_x$ is said Happened-Before an event block $v_y$ if $v_x$ is a (self-) ancestor of $v_y$. Denoted by $v_x \hbefore v_y$.}

The happened-before relation is the transitive closure of happens-immediately-before.
An event $v_x$ happened before an event $v_y$ if one of the followings happens: (a) $v_y \eself v_x$, (b) $v_y \eref v_x$,  or (c) $v_y \eancestor v_x$.

In Lachesis, we come up with the following proposition:
\begin{prop}[Happened-Immediate-Before OPERA]
	$v_x \hibefore v_y$ iff $v_y \erefz v_x$ iff edge $(v_y, v_x)$ $\in E$ in OPERA DAG.
\end{prop}
\begin{lem}[Happened-Before Lemma]
	$v_x \hbefore v_y$ iff $v_y \eancestor v_x$.
\end{lem}

\dfnn{Concurrent}{Two event blocks $v_x$ and $v_y$ are said concurrent if neither of them  happened before the other. Denoted by $v_x \concur v_y$.}

Let $G_1$ and $G_2$ be the two OPERA DAGS of any two honest nodes.
For any two vertices $v_x$ and $v_y$, if both of them are contained in two OPERA DAGs $G_1$ and $G_2$, then they are satisfied the following:
\begin{itemize}
	\item $v_x \hbefore v_y$ in $G_1$ if $v_x \hbefore v_y$ in $G_2$.
	\item $v_x \concur v_y$ in $G_1$ if $v_x \concur v_y$ in $G_2$.
\end{itemize}

Happened-before defines the relationship between event blocks created and shared between nodes. If there is a path from an event block $v_x$ to $v_y$, then $v_x$ Happened-before $v_y$. ``$v_x$ Happened-before $v_y$" means that the node creating $v_y$ knows event block $v_x$. This relation is the transitive closure of happens-immediately-before. Thus, an event $v_x$ happened before an event $v_y$ if one of the followings happens: (a) $v_y \eself v_x$, (b) $v_y \eref v_x$,  or (c) $v_y \eancestor v_x$. The happened-before relation of events form an acyclic directed graph $G' = (V, E')$ such that an edge $(v_i,v_j) \in E'$ has a reverse direction of the same edge in $E$.

\subsubsection{Lamport timestamp}

Lachesis protocol relies on Lamport timestamps to define a topological ordering of event blocks in OPERA DAG.
By using Lamport timestamps, we do not rely on physical clocks to determine a partial ordering of events.

We use this total ordering in our Lachesis protocol to determine consensus time.

For an arbitrary total ordering `$\prec$` of the processes, a relation `$\Rightarrow$` is defined as follows: if $v$ is an event in process $P_i$ and $v'$ is an event in process $P_j$, then $v \Rightarrow v'$ if and only if either (i) $C_i(v) < C_j(v')$ or (ii) $C(v)= C_j(v')$ and $P_i \prec P_j$. This defines a total ordering, and that the Clock Condition implies that if $v \rightarrow v'$ then $v \Rightarrow v'$.

\subsubsection{Domination relation}

In a graph $G=(V, E, r)$, a dominator is the relation between two vertices. A vertex $v$ is dominated by another vertex $w$, if every path in the graph from the root $r$ to $v$ have to go through $w$. The immediate dominator for a vertex $v$ is the last of $v$’s dominators, which every path in the graph have to go through to reach $v$.

\dfnn{Pseudo top}{A pseudo vertex, called top, is the parent of all top event blocks. Denoted by $\top$.}

\dfnn{Pseudo bottom}{A pseudo vertex, called bottom, is the child of all leaf event blocks. Denoted by $\bot$.}

With the pseudo vertices, we have $\bot$ happened-before all event blocks. Also all event blocks happened-before $\top$. That is, for all event $v_i$, $\bot \hbefore v_i$ and $v_i \hbefore \top$.

Then we define the domination relation for event blocks. To begin with, we first introduce pseudo vertices, \emph{top} and \emph{bot}, of the DAG OPERA DAG $G$.
\dfnn{pseudo top}{A pseudo vertex, called top, is the parent of all top event blocks. Denoted by $\top$.}

\dfnn{dom}{An event $v_d$ dominates an event $v_x$ if every path from $\top$ to $v_x$ must go through $v_d$. Denoted by $v_d \dom v_x$.}

\dfnn{strict dom}{An event $v_d$ strictly dominates an event $v_x$ if $v_d \dom v_x$ and $v_d$ does not equal $v_x$. Denoted by $v_d \sdom v_x$.}

\dfnn{domfront}{A vertex $v_d$ is said ``domfront'' a vertex $v_x$ if  $v_d$ dominates an immediate predecessor of $v_x$, but $v_d$ does not strictly dominate $v_x$. Denoted by $v_d \domf v_x$.}

\dfnn{dominance frontier}{The dominance frontier of a vertex $v_d$ is the set of all nodes $v_x$ such that $v_d \domf v_x$. Denoted by $DF(v_d)$.}

From the above definitions of domfront and dominance frontier, the following holds. If $v_d \domf v_x$, then $v_x \in DF(v_d)$.

For a set $S$ of vertices, an event $v_d$  $\frac{2}{3}$-dominates $S$ if there are more than 2/3 of vertices $v_x$ in $S$ such that $v_d$ dominates $v_x$. 	
Recall that $R_1$ is the set of all leaf vertices in $G$. The $\frac{2}{3}$-dom set $D_0$ is the same as the set $R_1$.The $\frac{2}{3}$-dom set $D_i$ is defined as follows:

\dfnn{$\frac{2}{3}$-dom set]}{A vertex $v_d$ belongs to a $\frac{2}{3}$-dom set  within the graph $G[v_d]$, if $v_d$ $\frac{2}{3}$-dominates $R_1$.
	The $\frac{2}{3}$-dom set $D_k$ consists of all roots $d_i$ such that  $d_i$ $\not \in $ $D_i$, $\forall$ $i$ = 1..($k$-1), and $d_i$ $\frac{2}{3}$-dominates $D_{i-1}$.}

\begin{lem}
	The $\frac{2}{3}$-dom set $D_i$ is the same with the root set $R_i$, for all nodes.
\end{lem}

\subsection{Proof of aBFT}\label{se:proof}
This section presents a proof to show that our Lachesis protocol is Byzantine fault tolerant when at most one-third of participants are compromised. We first provide some definitions, lemmas and theorems. We give formal semantics of PoS+DAG Lachesis protocol using the semantics in concurrent common knowledge.

\subsubsection{S-OPERA DAG - A Weighted DAG}

The OPERA DAG (DAG) is the local view of the DAG held by each node, this local view is used to identify topological ordering between events, to decide Clotho candidates and to compute consensus time of Atropos and events under it's subgraph.

\begin{defn}[S-OPERA DAG]
	An S-OPERA DAG is the local view of a weighted DAG $G$=($V$,$E$, $w$). Each vertex $v_i \in V$ is an event block. Each block has a weight, which is the \emph{validation score}. An edge ($v_i$,$v_j$) $\in E$ refers to a hashing reference from $v_i$ to $v_j$; that is, $v_i \erefz v_j$.
\end{defn}
	
\subsubsection{Consistency of DAGs}
	
\dfnn{Leaf}{The first created event block of a node is called a leaf event block.}

\dfnn{Root}{An event block $v$ is a root if either (1) it is the leaf event block of a node, or (2) $v$ can reach more than $2W/3$ validating power from previous roots.}

\dfnn{Root set}{The set of all first event blocks (leaf events) of all nodes form the first root set $R_1$ ($|R_1|$ = $n$). The root set $R_k$ consists of all roots $r_i$ such that $r_i$ $\not \in $ $R_i$, $\forall$ $i$ = 1..($k$-1) and $r_i$ can reach more than $2W/3$ validating power from other roots in the current frame, $i$ = 1..($k$-1).}

\dfnn{Frame}{The history of events are divided into frames. Each frame contains a disjoint set of roots and event blocks.}

\dfnn{Subgraph}{For a vertex $v$ in a DAG $G$, let $G[v] = (V_v,E_v)$ denote an induced-subgraph of $G$ such that $V_v$ consists of all ancestors of $v$ including $v$, and $E_v$ is the induced edges of $V_v$ in $G$.}

We define a definition of consistent DAGs, which is important to define the semantics of Lachesis protocol.

\begin{defn}[Consistent DAGs]
	Two OPERA DAGs $G_1$ and $G_2$ are consistent if for any event $v$ contained in both chains, $G_1[v] = G_2[v]$. Denoted by $G_1 \sim G_2$.
\end{defn}

\begin{thm}[Honest nodes have consistent DAGs]\label{thm:con-dags}
All honest nodes have consistent OPERA DAGs.
\end{thm}
	
If two nodes have OPERA DAGs containing event $v$, $v$ is valid and both contain all the parents referenced by $v$. In Lachesis, a node will not accept an event during a sync unless that node already has all references for that event, and thus both OPERA DAGs must contain $k$ references for $v$. The cryptographic hashes are assumed to be secure, therefore the references must be the same. By induction, all ancestors of $v$ must be the same. Therefore, for any two honest nodes, their OPERA DAGs are consistent. Thus, all honest nodes have consistent OPERA DAGs.\\

\dfnn{Creator}{If a node $n_x$ creates an event block $v$, then the creator of $v$, denoted by $cr(v)$, is $n_x$.}

\begin{defn}[Global DAG]
	A DAG $G^C$ is a global DAG of all $G_i$, if $G^C \sim G_i$ for all $G_i$.
\end{defn}

Let denote $G \sqsubseteq G'$ to stand for $G$ is a subgraph of $G'$. Some properties of the global DAG $G^C$ are given as follows:

\begin{enumerate}
	\item $\forall G_i$ ($G^C \sqsubseteq G_i$).
	\item
	$\forall v \in G^C$ $\forall G_i$ ($G^C[v] \sqsubseteq G_i[v]$).
	\item
	($\forall v_c \in G^C$) ($\forall v_p \in G_i$) (($v_p \hbefore v_c) \Rightarrow v_p \in G^C$).
\end{enumerate}

\subsubsection{Root}\label{pf:root}

We define the 'forkless cause' relation, as follows:
\begin{defn}[forkless cause]
The relation $forklessCause(A,B)$ denote that event $x$ is forkless caused by event $y$.	
\end{defn}

It means, in the subgraph of event $x$, $x$ did not observe any forks from $y$’s creator and at least QUORUM non-cheater validators have observed event $y$.

\begin{prop}
	If $G_1 \sim G_2$, and $x$ and $y$ exist in both, then forklessCause($x$, $y$) is true in $G_1$ iff forklessCause($x$, $y$) is true in $G_2$.
\end{prop}

For any two consistent graphs $G_1$ and $G_2$, the subgraphs $G_1[x]=G_2[x]$ and $G_1[y]=G_2[y]$. Also $y$'s creator and QUORUM is the same on both DAGs.

\dfnn{Validation Score}{For event block $v$ in both $G_1$ and $G_2$, and $G_1 \sim G_2$, the validation score of $v$ in $G_1$ is identical with that of $v$ in $G_2$.}

\begin{defn}[root]
An event is called a root, if it is \emph{forkless causes} by QUORUM roots of previous frame.
\end{defn}

For every event, its frame number is no smaller than self-parent's frame number. 

The lowest possible frame number is 1. Relation $forklessCause$ is a stricter version of \texttt{happened-before} relation, i.e. if $y$ forkless causes $x$, then $y$ is happened before A. 
Yet, unlike \texttt{happened-before} relation, \texttt{forklessCause} is not transitive because it returns false if fork is observed. So if $x$ forkless causes $y$, and $y$ forkless causes $z$, it does not mean that $x$ forkless causes $z$. If there's no forks in the graph, then the \texttt{forklessCause} is always transitive.

\begin{prop}
	If $G_1 \sim G_2$, and  then $G_1$ and $G_2$ are root consistent.
\end{prop}

Now we state the following important propositions.
\begin{defn}[Root consistency]
	Two DAGs $G_1$ and $G_2$ are root consistent: if for every $v$ contained in both DAGs, and $v$ is a root of $j$-th frame in $G_1$, then $v$ is a root of $j$-th frame in $G_2$.
\end{defn}
\begin{prop}
	If $G_1 \sim G_2$, then $G_1$ and $G_2$ are root consistent.
\end{prop}
\begin{proof}
	By consistent chains, if $G_1 \sim G_2$ and $v$ belongs to both chains, then $G_1[v]$ = $G_2[v]$.
	We can prove the proposition by induction. For $j$ = 0, the first root set is the same in both $G_1$ and $G_2$. Hence, it holds for $j$ = 0. Suppose that the proposition holds for every $j$ from 0 to $k$. We prove that it also holds for $j$= $k$ + 1.
	Suppose that $v$ is a root of frame $f_{k+1}$ in $G_1$. 
	Then there exists a set $S$ reaching 2/3 of members in $G_1$ of frame $f_k$ such that $\forall u \in S$ ($u\hbefore v$). As $G_1 \sim G_2$, and $v$ in $G_2$, then $\forall u \in S$ ($u \in G_2$). Since the proposition holds for $j$=$k$, 
	As $u$ is a root of frame $f_{k}$ in $G_1$, $u$ is a root of frame $f_k$ in $G_2$. Hence, the set $S$ of 2/3 members $u$ happens before $v$ in $G_2$. So $v$ belongs to $f_{k+1}$ in $G_2$. The proposition is proved.
\end{proof}

From the above proposition, one can deduce the following:
\begin{lem}
	$G^C$ is root consistent with $G_i$ for all nodes.
\end{lem}

By consistent DAGs, if $G_1 \sim G_2$ and $v$ belongs to both chains, then $G_1[v]$ = $G_2[v]$.
We can prove the proposition by induction. For $j$ = 0, the first root set is the same in both $G_1$ and $G_2$. Hence, it holds for $j$ = 0. Suppose that the proposition holds for every $j$ from 0 to $k$. We prove that it also holds for $j$= $k$ + 1.
Suppose that $v$ is a root of frame $f_{k+1}$ in $G_1$. 
Then there exists a set $S$ reaching 2/3 of members in $G_1$ of frame $f_k$ such that $\forall u \in S$ ($u\hbefore v$). As $G_1 \sim G_2$, and $v$ in $G_2$, then $\forall u \in S$ ($u \in G_2$). Since the proposition holds for $j$=$k$, 
As $u$ is a root of frame $f_{k}$ in $G_1$, $u$ is a root of frame $f_k$ in $G_2$. Hence, the set $S$ of 2/3 members $u$ happens before $v$ in $G_2$. So $v$ belongs to $f_{k+1}$ in $G_2$.

Thus, all honest nodes have the same consistent root sets, which are the root sets in $G^C$. Frame numbers are consistent for all nodes.

\subsubsection{Clotho}\label{pf:clotho}

Algorithm~\ref{al:selClotho} shows our Lachesis's algorithm to decide a Clotho candidate, as given in Section~\ref{se:Clotho-sel}.

\begin{defn}[Clotho consistency]
Two DAGs $G_1$ and $G_2$ are Clotho consistent: if every $v$ contained in both DAGs, and $v$ is a Clotho candidate in $j$-th frame in $G_1$, then $v$ is a Clotho candidate in $j$-th frame in $G_2$.
\end{defn}

Thus, we have the following proposition.	
\begin{prop}
	If $G_1 \sim G_2$, and  then $G_1$ and $G_2$ are Clotho consistent.
\end{prop}

The Algorithm~\ref{al:selClotho} uses computed roots, and $forklessCause$ relation to determine which roots are candidates for Clotho.
As proved in previous proposition and lemma, honest nodes have consistent DAGs, consistent roots, and so is the $forklessCause$ relation.
Therefore, the Clotho candidate selection is consistent across honest nodes.
All honest nodes will decide the same set of Clotho events from the global DAG $G^C$.

\subsubsection{Atropos}\label{pf:atropos}

Algorithm for deciding Atropos is shown in Algorithm~\ref{al:acs} in Section~\ref{se:Atropos-sel}. 

The set of roots and Clothos are consistent in all honest nodes, as proved in previous sections. Validators and their stakes are unchanged within an epoch.
The selection of first Atropos from the sorted list of Clotho is deterministic. So Atropos candidate is consistent.

\begin{defn}[Atropos consistency]
	Two DAGs $G_1$ and $G_2$ are Atropos consistent: if every $v$ contained in both DAGs, and $v$ is an Atropos in $j$-th frame in $G_1$, then $v$ is an Atropos in $j$-th frame in $G_2$.
\end{defn}
	
Thus, we have the following proposition.	
\begin{prop}
	If $G_1 \sim G_2$, then $G_1$ and $G_2$ are Atropos consistent.
\end{prop}	

All honest nodes will decide the same set of Atropos events from the global DAG $G^C$.

\subsubsection{Block}\label{pf:block}

After a new Atropos is decided, the Atropos is assigned a consensus time using the Atropos's median timestamp. The median timestamp is the median of the times sent by all nodes.
The median time makes sure that the consensus time is BFT, even in case there may exist Byzantine nodes accounting for up to 1/3 of the validation power.

When an Atropos is decided and assigned a consensus time, a new block is created. The events under the Atropos $a_i$ but not included yet on previous blocks will be sorted using a topological sorting algorithms outlined in Section~\ref{se:blocks} and ~\ref{se:toposort}. The subgraph $G[a_i]$ under the Atropos $a_i$ is consistent (as proved previous lemma). The sorting algorithms are deterministic, and thus the sorted list of events will be the same on all honest nodes.

After sorting, all these events under the new Atropos will be added into the new block, and they are all assigned the same consensus time, which is the Atropos's consensus time. Transactions of these events are added into the new block in the sorted order.

\begin{defn}[Block Consistency]
	Two DAGs $G_1$ and $G_2$ are Block consistent: if a block $b_i$ at position $i$ will consist of the same sequence of event blocks and transactions, for every $i$.
\end{defn}

From the above, event blocks in the subgraph rooted at the Atropos are also final events. 
It leads to the following proposition:
\begin{prop}
	If $G_1 \sim G_2$, then $G_1$ and $G_2$ are Block consistent.
\end{prop}

Let $\prec$ denote an arbitrary total ordering  of the nodes (processes) $p_i$ and $p_j$. 	

\dfnn{Total ordering}{\emph{Total ordering} is a relation $\Rightarrow$ satisfying the following: for any event $v_i$ in $p_i$ and any event $v_j$ in $p_j$, $v_i \Rightarrow v_j$ if and only if either (i) $C_i(v_i) < C_j(v_j)$ or (ii) $C_i(v_i)$=$C_j(v_j)$ and $p_i \prec p_j$.}

This defines a total ordering relation. The Clock Condition implies that if $v_i \rightarrow v_j$ then $v_i \Rightarrow v_j$. 

\subsubsection{Main chain}

The Main chain consists of new blocks, each of which is created based on a new Atropos.

\begin{defn}[Consistent Chain]
	Two DAGs $G_1$ and $G_2$ are Chain consistent: if $G_1$ and $G_2$ are Block consistent and they have the same number of blocks.
\end{defn}
\begin{prop}
	All honest nodes are Chain consistent.
\end{prop}

\begin{thm}[Finality]
	\label{thm:finality}
	All honest nodes produce the same ordering and consensus time for finalized event blocks.
\end{thm}
\begin{proof}
	As proved in the above Section~\ref{pf:root}, the set of roots is consistent across the nodes. From section~\ref{pf:clotho}, the set of Clotho events is consistent across the nodes. In Section~\ref{pf:atropos}, our topological sorting of Clotho, which uses validator'stake, Lamport's timestamp is deterministic, and so it will give the same ordering of Atropos events. The Atropos's timestamp is consistent between the nodes, as proved in Section~\ref{pf:atropos}. Section~\ref{pf:block} outlines a proof that the topology ordering of vertices in each subgraph under an Atropos is the same across the nodes. Also finalized event blocks in each block is assigned the same consensus time, which is the Atropos's consensus timestamp.
\end{proof}

\subsubsection{Fork free}

\begin{defn}[Fork]
	The pair of events ($v_x$, $v_y$) is a fork if $v_x$ and $v_y$ have the same creator, but neither is a self-ancestor of the other. Denoted by $v_x \efork v_y$.
\end{defn}
For example, let $v_z$ be an event in node $n_1$ and two child events $v_x$ and $v_y$ of $v_z$. if $v_x \eself v_z$, $v_y \eself v_z$, $v_x \not \eself v_y$, $v_y \not \eself v_z$, then ($v_x$, $v_y$) is a fork.
The fork relation is symmetric; that is $v_x \efork v_y$ iff $v_y \efork v_x$.
\begin{lem}
	$v_x \efork v_y$ iff $cr(v_x)=cr(v_y)$ and $v_x \concur v_y$.
\end{lem}
\begin{proof}
	By definition, ($v_x$, $v_y$) is a fork if $cr(v_x)=cr(v_y)$, $v_x \not \eancestor v_y$ and $v_y \not \eancestor v_x$. Using Happened-Before, the second part means $v_x \not \rightarrow v_y$ and $v_y \not \rightarrow v_x$. By definition of concurrent, we get $v_x \concur v_y$.
\end{proof}

\begin{lem} (Fork detection). If there is a fork $v_x \efork  v_y$, then $v_x$ and $v_y$ cannot both be roots on honest nodes.
\end{lem}
\begin{proof}
	Here, we show a proof by contradiction. Any honest node cannot accept a fork so $v_x$ and $v_y$ cannot be roots on the same honest node. Now we prove a more general case. Suppose that both $v_x$ is a root on node $n_x$ and $v_y$ is root on node $n_y$, where $n_x$ and $n_y$ are honest nodes. Since $v_x$ is a root, it reached events created by more than $2/3$ of total validating power. Similarly, $v_y$ is a root, it reached events created by more than $2W/3$. Thus, there must be an overlap of more than $W$/3 members of those events in both sets. Since Lachesis protocol is 1/3-aBFT, we assume less than  1/3 of total validating power are from malicious nodes, so there must be at least one honest member in the overlap set. Let $n_m$ be such an honest member. Because $n_m$ is honest, $n_m$ can detect and it does not allow the fork. This contradicts the assumption. Thus, the lemma is proved.
\end{proof}

\begin{prop}[Fork-free]
	For any Clotho $c$, the subgraph $G[c]$ is fork-free.
\end{prop}
\begin{lem}[Fork free lemma]\label{lem:fork}
	If there is a fork $v_x \efork v_y$, this fork will be detected by a Clotho.
\end{lem}

\begin{proof}
	Suppose that a node creates two event blocks ($v_x, v_y$), which forms a fork. Let's assume there exist two Clotho $c_x$ and $c_y$ that can reach both events. From Algorithm~\ref{al:selClotho}, each Clotho must reach more than $2W/3$. Thus, there must more than $W/3$ of honest nodes would know both $c_x$ and $c_y$. Lachesis protocol assumes that less than $W/3$ are from malicious nodes. Therefore, there must exist at least a node that knows both Clothos $c_x$ and $c_y$. Then that node must know both $v_x$ and $v_y$. So the node will detect the fork. This is a contradiction. The lemma is proved.
\end{proof}

\begin{thm}[Fork Absence]\label{thm:same}
	All honest nodes build the same global DAG $G^C$, which contains no fork.
\end{thm}
\begin{proof}
	Suppose that there are two event blocks $v_x$ and $v_y$ contained in both $G_1$ and $G_2$, and their path between $v_x$ and $v_y$ in $G_1$ is not equal to that in $G_2$. We can consider that the path difference between the nodes is a kind of fork attack. Based on Lemma~\ref{lem:fork}, if an attacker forks an event block, each chain of $G_i$ and $G_2$ can detect and remove the fork before the Clotho is generated. Thus, any two nodes have consistent OPERA DAG. 
\end{proof}

\subsection{Semantics of Lachesis}\label{sec:semantics} 

This section gives the formal semantics of Lachesis consensus protocol.
We use CCK model \cite{cck92} of an asynchronous system as the base of the semantics of our Lachesis protocol.
We present notations and concepts, which are important for Lachesis protocol. We adapt the notations and concepts of CCK paper to suit our Lachesis protocol.
Events are ordered based on Lamport's happened-before relation. 
In particular, we use Lamport’s theory to describe global states of an asynchronous system.
 
\subsubsection{Node State}
A node is a machine participating in the Lachesis protocol.
Each node has a local state consisting of local histories, messages, event blocks, and peer information.

\dfnn{State}{A (local) state of node $i$ is denoted by $s_j^i$ consisting of a sequence of event blocks $s_j^i$=$v_0^i$, $v_1^i$, $\dots$, $v_j^i$.}

In a DAG-based protocol, each  event block $v_j^i$ is \emph{valid} only if the reference blocks exist before it. A local state $s_j^i$ is corresponding to a unique DAG. In \stakedag, we simply denote the $j$-th local state of a node $i$ by the DAG $g_j^i$. Let $G_i$ denote the current local DAG of a process $i$.

\dfnn{Action}{An action is a function from one local state to another local state.}

An action can be either: a $send(m)$ action of a message $m$, a $receive(m)$ action, and an internal action. A message $m$ is a triple $\langle i,j,B \rangle$ where the sender $i$, the message recipient $j$, and the message body $B$. 
In \stakedag, $B$ consists of the content of an event block $v$. Let $M$ denote the set of messages.
Semantics-wise, there are two actions that can change a process's local state: creating a new event and receiving an event from another process.

\dfnn{Event}{An event is a tuple $\langle  s,\alpha,s' \rangle$ consisting of a state, an action, and a state.}

Sometimes, the event can be represented by the end state $s'$. 
The $j$-th event in history $h_i$ of process $i$ is $\langle  s_{j-1}^i,\alpha,s_j^i \rangle$, denoted by $v_j^i$.

\dfnn{Local history}{A local history $h_i$ of $i$ is a sequence of local states starting with an initial state. A set $H_i$ of possible local histories for each process $i$.}

A process's state can be obtained from its initial state and the sequence of actions or events that have occurred up to the current state. Lachesis protocol uses append-only semantics. The local history may be equivalently described as either of the followings: 
$$h_i = s_0^i,\alpha_1^i,\alpha_2^i, \alpha_3^i \dots $$
$$h_i = s_0^i, v_1^i,v_2^i, v_3^i \dots $$
$$h_i = s_0^i, s_1^i, s_2^i, s_3^i, \dots$$

In Lachesis,, a local history is equivalently expressed as:
$$h_i = g_0^i, g_1^i, g_2^i, g_3^i, \dots$$
where $g_j^i$ is the $j$-th local DAG (local state) of the process $i$.

\dfnn{Run}{Each asynchronous run is a vector of local histories. Denoted by
	$\sigma$ = $\langle h_1,h_2,h_3,...h_N \rangle$.}

Let $\Sigma$ denote the set of asynchronous runs. A global state of run $\sigma$ is an $n$-vector of prefixes of local histories of $\sigma$, one prefix per process. The happened-before relation can be used to define a consistent global state, often termed a consistent cut, as follows.

\subsubsection{Consistent Cut} 

An asynchronous system consists of the following sets: a set $P$ of process identifiers; a set $C$ of channels; a set $H_i$ is the set of possible local histories for each process $i$; a set $A$ of asynchronous runs; a set $M$ of all messages.  If node $i$ selects node $j$ as a peer, then $(i,j) \in C$. In general, one can express the history of each node in DAG-based protocol in general or in Lachesis protocol in particular, in the same manner as in the CCK paper~\cite{cck92}.

We can now use Lamport’s theory to talk about global states of an asynchronous system.
A global state of run $\sigma$ is an $n$-vector of prefixes of local histories of $\sigma$, one prefix per process.
The happened-before relation can be used to define a consistent global state, often termed a consistent cut, as follows.

\dfnn{Consistent cut}{A consistent cut of a run $\sigma$ is any global state such that if $v_x^i \rightarrow v_y^j$ and $v_y^j$ is in the global state, then $v_x^i$ is also in the global state. Denoted by $\vec{c}(\sigma)$.}

By Theorem~\ref{thm:con-dags}, all nodes have consistent local OPERA DAGs. The concept of consistent cut formalizes such a global state of a run. Consistent cuts represent the concept of scalar time in distributed computation, it is possible to distinguish between a ``before'' and an ``after''. In Lachesis, a consistent cut consists of all consistent OPERA DAGs. A received event block exists in the global state implies the existence of the original event block.
Note that a consistent cut is simply a vector of local states; we will use the notation $\vec{c}(\sigma)[i]$ to indicate the local state of $i$ in cut $\vec{c}$ of run $\sigma$.

The formal semantics of an asynchronous system is given via  the satisfaction relation $\vdash$. Intuitively $\vec{c}(\sigma) \vdash \phi$, ``$\vec{c}(\sigma)$ satisfies $\phi$,'' if fact $\phi$ is true in cut $\vec{c}$ of run $\sigma$. 
We assume that we are given a function $\pi$ that assigns a truth value to each primitive proposition $p$. The truth of a primitive proposition $p$ in $\vec{c}(\sigma)$ is determined by $\pi$ and $\vec{c}$. This defines $\vec{c}(\sigma) \vdash p$.

\begin{defn}[Equivalent cuts]
	Two cuts $\vec{c}(\sigma)$ and $\vec{c'}(\sigma')$ are equivalent  with respect to $i$ if: $$\vec{c}(\sigma) \sim_i \vec{c'}(\sigma') \Leftrightarrow \vec{c}(\sigma)[i] = \vec{c'}(\sigma')[i]$$
\end{defn}

A message chain of an asynchronous run is a sequence of messages $m_1$, $m_2$, $m_3$, $\dots$, such that, for all $i$, $receive(m_i)$ $\rightarrow$  $send(m_{i+1})$. Consequently, $send(m_1)$ $\rightarrow$ $receive(m_1)$ $\rightarrow$ $send(m_2)$ $\rightarrow$ $receive(m_2)$ $\rightarrow$ $send(m_3)$ $\dots$.

We introduce two families of modal operators, denoted by $K_i$ and $P_i$, respectively. Each family indexed by process identifiers. 
Given a fact $\phi$, the modal operators are defined as follows:

\dfnn{$i$ knows $\phi$}{$K_i(\phi)$ represents the statement ``$\phi$ is true in all possible consistent global states that include $i$’s local state''. 
	$$\vec{c}(\sigma) \vdash K_i(\phi) \Leftrightarrow \forall \vec{c'}(\sigma')   (\vec{c'}(\sigma') \sim_i \vec{c}(\sigma) \ \Rightarrow\ \vec{c'}(\sigma') \vdash \phi) $$}

\dfnn{$i$ partially knows $\phi$}{$P_i(\phi)$ represents the statement ``there is some consistent global state in this run that includes $i$’s local state, in which $\phi$ is true.''
	$$\vec{c}(\sigma) \vdash P_i(\phi) \Leftrightarrow \exists \vec{c'}(\sigma) ( \vec{c'}(\sigma) \sim_i \vec{c}(\sigma) \ \wedge\ \vec{c'}(\sigma) \vdash \phi )$$}

The last modal operator is concurrent common knowledge (CCK), denoted by $C^C$.
\begin{defn}[Concurrent common knowledge]
	$C^C(\phi)$ is defined as a fixed point of $M^C(\phi \wedge X)$
\end{defn}
CCK defines a state of process knowledge that implies that all processes are in that same state of knowledge, with respect to $\phi$, along some cut of the run. In other words, we want a state of knowledge $X$ satisfying: $X = M^C(\phi \wedge X)$.	
$C^C$ will be defined semantically as the weakest such fixed point, namely as the greatest fixed-point of $M^C(\phi \wedge X)$.
It therefore satisfies:
$$C^C(\phi) \Leftrightarrow  M^C(\phi \wedge C^C(\phi))$$

Thus, $P_i(\phi)$ states that there is some cut in the same asynchronous run $\sigma$ including $i$’s local state, such that $\phi$ is true in that cut.\\

Note that $\phi$ implies $P_i(\phi)$. But it is not the case, in general, that $P_i(\phi)$ implies $\phi$ or even that $M^C(\phi)$ implies $\phi$. The truth of $M^C(\phi)$ is determined with respect to some cut $\vec{c}(\sigma)$. A process cannot distinguish which cut, of the perhaps many cuts that are in the run and consistent with its local state, satisfies $\phi$; it can only know the existence of such a cut.\\ 

\dfnn{Global fact}{Fact $\phi$ is valid in system $\Sigma$, denoted by $\Sigma \vdash \phi$, if $\phi$ is true in all cuts of all runs of $\Sigma$.
	$$\Sigma \vdash \phi 
	\Leftrightarrow (\forall \sigma \in \Sigma)(\forall\vec{c}) (\vec{c}(a) \vdash \phi)$$}

\dfnn{Valid Fact}{Fact $\phi$ is valid, denoted $\vdash \phi$, if $\phi$ is valid in all systems, i.e. $(\forall \Sigma) (\Sigma \vdash \phi)$.
}

\dfnn{Local fact}{A fact $\phi$ is local to process $i$ in system $\Sigma$ if
	$\Sigma \vdash (\phi \Rightarrow K_i \phi)$.}

\subsubsection{State in aBFT Lachesis}

We introduce a new modal operator, written as $M^C$, which stands for ``majority concurrently knows''. The definition of $M^C(\phi)$ is as follows.

\begin{defn}[Majority concurrently knows]
	$$M^C(\phi) =_{def} \bigwedge_{i \in S} K_i P_i(\phi), $$ where $S \subseteq P$ and $|S| > 2n/3$.	
\end{defn}

This is adapted from the ``everyone concurrently knows'' in CCK paper~\cite{cck92}.
In the presence of one-third of faulty nodes, the original operator ``everyone concurrently knows'' is sometimes not feasible.
Our modal operator $M^C(\phi)$ fits precisely the semantics for BFT systems, in which unreliable processes may exist.\\

A weighted version of ``majority concurrently knows'' is ``Quorum concurrently knows'', which is defined as follows:

\begin{defn}[Quorum concurrently knows]
	$$Q^C(\phi) =_{def} \bigwedge_{i \in S} K_i P_i(\phi), $$ where $S \subseteq P$ and $w(S) > 2W/3$.
\end{defn}

The total weight of a set is equal to the sum of the weights of each node in the set; that is, $w(S) = \sum_{1}^{|S|}{w(n_i)}$.

\begin{thm} If $\phi$ is local to process $i$ in system $\Sigma$, then $\Sigma \vdash (P_i(\phi) \Rightarrow \phi)$.	
\end{thm}

\begin{lem}
	If fact $\phi$ is local to 2/3 of the processes in a system $\Sigma$, then $\Sigma \vdash (M^C(\phi) \Rightarrow \phi)$ and furthermore $\Sigma \vdash (C^C(\phi) \Rightarrow \phi)$.
\end{lem}

\begin{defn}
	A fact $\phi$ is attained in run $\sigma$ if $\exists \vec{c}(\sigma) (\vec{c}(\sigma) \vdash \phi)$.
\end{defn}

Often, we refer to ``knowing'' a fact $\phi$ in a state rather than in a consistent cut, since knowledge is dependent only on the local state of a process.
Formally, $i$ knows $\phi$ in state $s$ is shorthand for
$$\forall \vec{c}(\sigma) (\vec{c}(\sigma)[i] = s \Rightarrow \vec{c}(\sigma) \vdash \phi)$$

For example, if a process in Lachesis protocol knows a fork exists (i.e., $\phi$ is the exsistenc of fork) in its local state $s$ (i.e., $g_j^i$), then a consistent cut contains the state $s$ will know the existence of that fork.

\begin{defn}[$i$ learns $\phi$]
	Process $i$ learns $\phi$ in state $s_j^i$ of run $\sigma$ if $i$ knows $\phi$ in $s_j^i$ and, for all previous states $s_k^i$ in run $\sigma$, $k < j$, $i$ does not know $\phi$.
\end{defn}

The following theorem says that if $C_C(\phi$ is attained in a run then all processes $i$ learn $P_i C^C(\phi)$ along a single consistent cut.

\begin{thm}[attainment]
	If $C^C(\phi)$ is attained in a run $\sigma$, then the set of states in which all processes learn $P_i C^C(\phi)$ forms a consistent cut in $\sigma$.	
\end{thm}

We have presented a formal semantics of Lachesis protocol based on the concepts and notations of concurrent common knowledge~\cite{cck92}.
For a proof of the above theorems and lemmas in this Section, we can use similar proofs as described in the original CCK paper.

With the formal semantics of Lachesis, the theorems and lemmas described in Section~\ref{se:proof} can be expressed in term of CCK. For example, one can study a fact $\phi$ (or some primitive proposition $p$) in the following forms: `is there any existence of fork?'. One can make Lachesis-specific questions like 'is event block $v$ a root?', 'is $v$ a clotho?', or 'is $v$ a atropos?'. This is a remarkable result, since we are the first that define such a formal semantics for DAG-based protocol.

\section{Alternatives for scalability}

\subsection{Layering-based Model}\label{se:Sprotocol}

In this section, we present an approach that uses graph layering on top of our PoS DAG-based Lachesis protocol. Our general model is given in \stakedag\ protocol~\cite{stakedag}, which is based on layering-based approach of ONLAY framework~\cite{onlay19}. After using layering algorithms on the OPERA DAG, the assigned layers are then used to determine consensus of the event blocks.

\subsubsection{Layering Definitions}\label{se:layering}

For a directed acyclic graph $G$=($V$,$E$), a layering is to assign a layer number to each vertex in $G$.

\dfnn{Layering}{A layering (or levelling) of $G$ is a topological numbering $\phi$ of $G$, $\phi: V \rightarrow Z$,  mapping the  set  of  vertices $V$ of $G$ to  integers  such  that $\phi(v)$ $\geq$ $\phi(u)$ + 1 for every directed edge ($u$, $v$) $\in E$.  If $\phi(v)$=$j$, then $v$ is a layer-$j$ vertex and $V_j= \phi^{-1}(j)$ is the jth layer of $G$.}

From happened before relation, $v_i \hbefore v_j$ iff $\phi(v_i) < \phi(v_j)$.
Layering $\phi$ partitions the set of vertices $V$ into a finite number $l$ of \emph{non-empty} disjoint subsets (called layers) $V_1$,$V_2$,$\dots$, $V_l$, such that $V$ = $\cup_{i=1}^{l}{V_i}$. Each vertex is assigned to a layer $V_j$, where $1\leq j \leq l$, such that every edge ($u$,$v$) $\in E$, $u \in V_i$, $v \in V_j$, $1 \leq i < j \leq l$.

\dfnn{Hierarchical graph}{
	For a layering $\phi$, 
	the produced graph $H$=($V$,$E$,$\phi$) is a \emph{hierarchical graph}, which is also called an $l$-layered directed graph and could be represented as
	$H$=($V_1$,$V_2$,$\dots$,$V_l$;$E$).}

Approaches to DAG layering include minimizing the height~\cite{bang2008digraphs,Sedgewick2011}, fixed width layering~\cite{Coffman1972} and minimizing the total edge span. 
A simple layering algorithm to achieve minimum height is as follows:
First, all source vertices are placed in the first layer $V_1$.   
Then, the layer $\phi(v)$ for every remaining vertex $v$ is recursively defined by $\phi(v)$ = $max\{\phi(u) | (u, v) \in E\}$ + 1.
This algorithm produces a layering where many vertices will stay close to the bottom, and hence the number of layers $k$ is kept minimized. By using a topological ordering of the vertices~\cite{Mehlhorn84a}, the algorithm can be implemented in linear time $O$($|V|$+$|E|$).

\subsubsection{H-OPERA DAG}

We introduce a new notion of H-OPERA DAG, which is built on top of the OPERA DAG by using graph layering.

\begin{defn}[H-OPERA DAG]
	An H-OPERA DAG is the result hierarchical graph $H = (V,E, \phi)$.
\end{defn}

H-OPERA DAG can also be represented by $H$=($V_1$,$V_2$,$\dots$,$V_l$;$E$). Vertices are assigned with layers such that each edge ($u$,$v$) $\in E$ flows from a higher layer to a lower layer i.e., $\phi(u)$ $>$ $\phi(v)$.

Generally speaking, a layering can produce a proper hierarchical graph by introducing dummy vertices along every long edge. However, for consensus protocol, such dummy vertices can increase the computational cost of the H-OPERA itself and of the following steps to determine consensus. Thus, in \onlay\ we consider layering algorithms that do not introduce dummy vertices.

A simple approach to compute H-OPERA DAG is to consider OPERA DAG as a static graph $G$ and then apply a layering algorithm such as either LongestPathLayer (LPL) algorithm or Coffman-Graham (CG) algorithm on the graph $G$.
LPL algorithm can achieve HG in linear time $O$($|V|$+$|E|$) with the minimum height. CG algorithm has a time complexity of $O$($|V|^2$), for a given $W$.

To handle large $G$ more efficiently, we introduce online layering algorithms. We adopt the online layering algorithms introduced in~\cite{onlay19}. Specifically, we use \emph{Online Longest Path Layering} (O-LPL) and \emph{Online Coffman-Graham} (O-CG) algorithm, which assign layers to the vertices in diff graph $G'$=($V'$,$E'$) consisting of new self events and received unknown events.
The algorithms are efficient and scalable to compute the layering of large dynamic S-OPERA DAG.
Event blocks from a new graph update are sorted in topological order before being assigned layering information.

\subsubsection{Root Graph}\label{se:rootgraph}

There is a Root selection algorithm given in Section~\ref{se:rootsel}.
Here, we present a new approach that uses root graph and frame assignment.

\begin{defn}[Root graph] 
	A root graph $G_R$=($V_R$, $E_R$) is a DAG consisting of vertices as roots and edges represent their reachability.
\end{defn}

In root graph $G_R$, the set of roots $V_R$ is a subset of $V$. Edge ($u$,$v$) $\in$ $E_R$ only if $u$ and $v$ are roots and $u$ can reach $v$ following edges in $E$ i.e., $v \hbefore u$. Root graph contains the $n$ genesis vertices i.e., the $n$ leaf event blocks. 
A vertex $v$ that can reach at least 2/3$W$ + 1 of the current set of all roots $V_R$ is itself a root. For each root $r_i$ that new root $r$ reaches, we include a new edge ($r$, $r_i$) into the set of root edges $E_R$.
Note that, if a root $r$ reaches two other roots $r_1$ and $r_2$ of the same node, and $\phi(r_1) >\phi(r_2)$ then we include only one edge ($r$, $r_1$) in $E_R$. This requirement ensures that each root of a node can have at most one edge to any other node.

\begin{algorithm}[htb]
	\caption{Root graph algorithm}\label{al:buildrootgraph}
	\begin{algorithmic}[1]
		\State Require: H-OPERA DAG
		\State Output: root graph $G_R=(V_R,E_R)$
		\State{$R \leftarrow$ set of leaf events}
		\State $V_R \leftarrow R$
		\State $E_R \leftarrow \emptyset$
		
		\Function{buildRootGraph}{$\phi$, $l$}
		\For{each layer $i$=1..$l$}
		\State$Z \leftarrow \emptyset$
		\For{each vertex $v$ in layer $\phi_i$}
		\State $S \leftarrow$ the set of vertices in $R$ that $v$ reaches
		\If{$w_S > 2/3W$}
		\For{each vertex $v_i$ in $S$}
		\State $E_R \leftarrow E_R \cup \{(v,v_i)\}$
		\EndFor
		\State $V_R \leftarrow V_R \cup \{v\}$
		\State $Z \leftarrow Z \cup \{v\}$ 
		\EndIf
		\EndFor
		
		\For{each vertex $v_j$ in $Z$}
		\State Let $v_{old}$ be a root in $R$ such that $cr(v_j) = cr(v_{old})$		
		\State $R \leftarrow R \setminus \{v_{old}\} \cup \{v_j\}$
		\EndFor
		\EndFor
		\EndFunction
	\end{algorithmic}
\end{algorithm}

Figure~\ref{fig:daglayerrootframe-3} shows an example of the H-OPERA DAG, which is the resulting hierarchical graph from a laying algorithm. Vertices of the same layer are placed on a horizontal line. There are 14 layers which are labelled as L0, L1, ..., to L14. Roots are shown in red colors.
Figure~\ref{fig:rootgraph-3} gives an example of the root graph constructed from the H-OPERA DAG.

\begin{figure}[htb]
	\centering			
	\subfloat[Hierarchical graph with frame assignment]{\label{fig:daglayerrootframe-3}
		\includegraphics[width=0.25\linewidth]{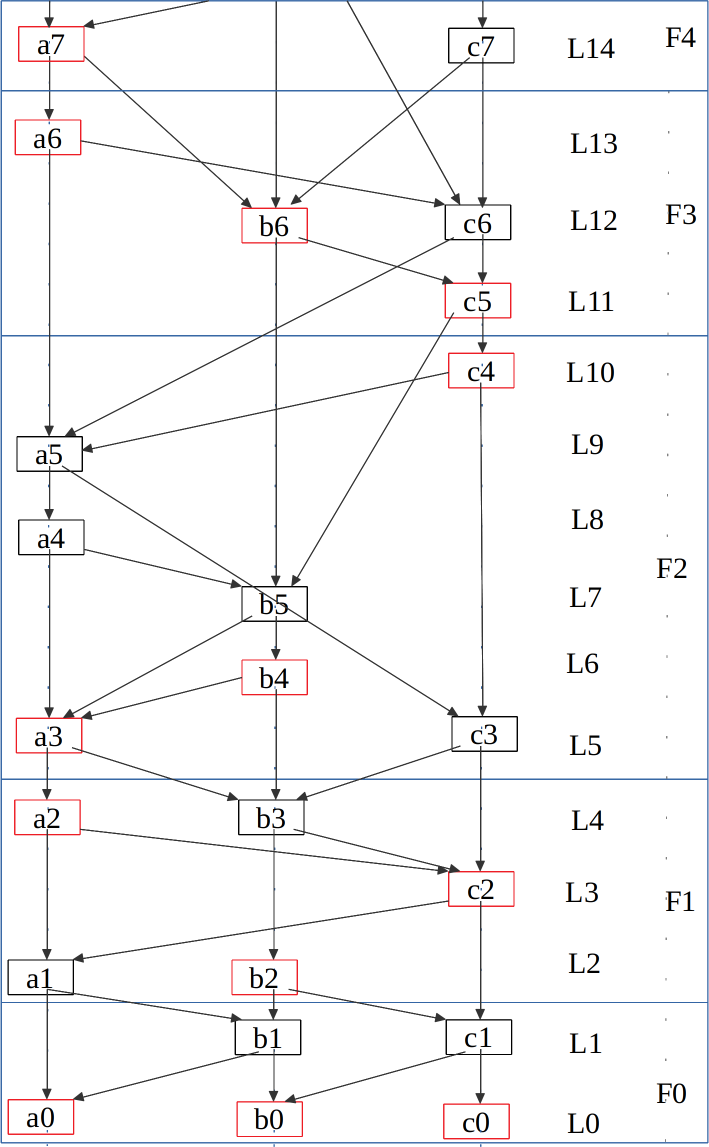}}
	\qquad
	\subfloat[Root graph]{\label{fig:rootgraph-3}\includegraphics[width=0.25\linewidth]{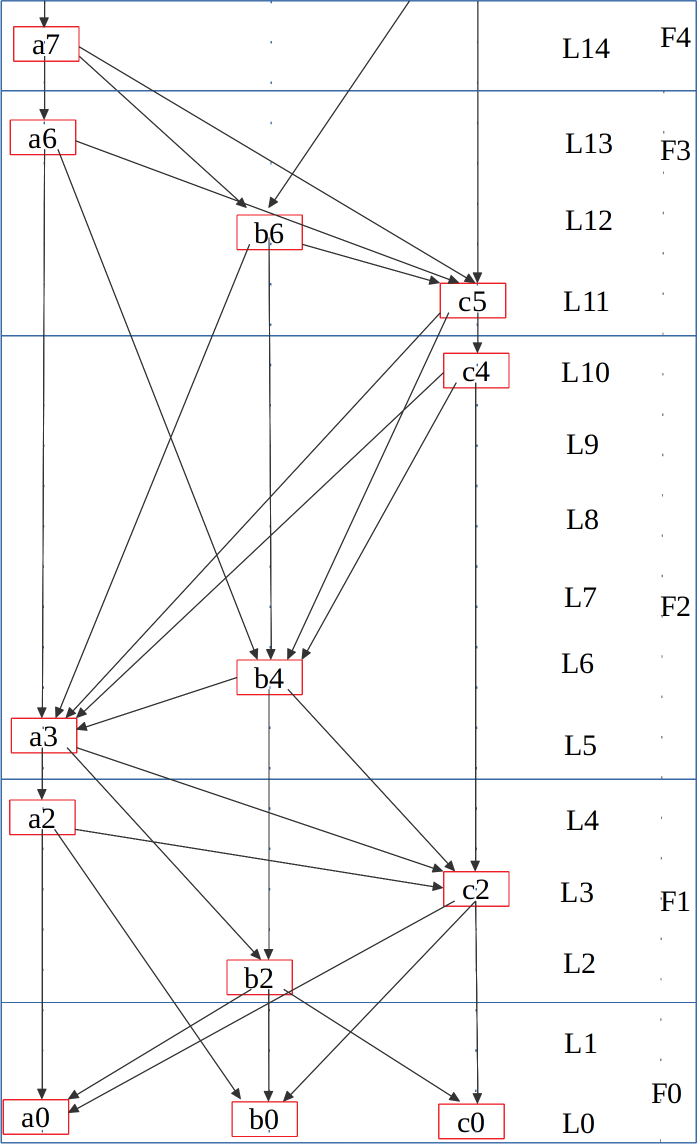}}
	\qquad
	\subfloat[Root-layering of a root graph]{\label{fig:rootgraph-3-layer}\includegraphics[width=0.25\linewidth]{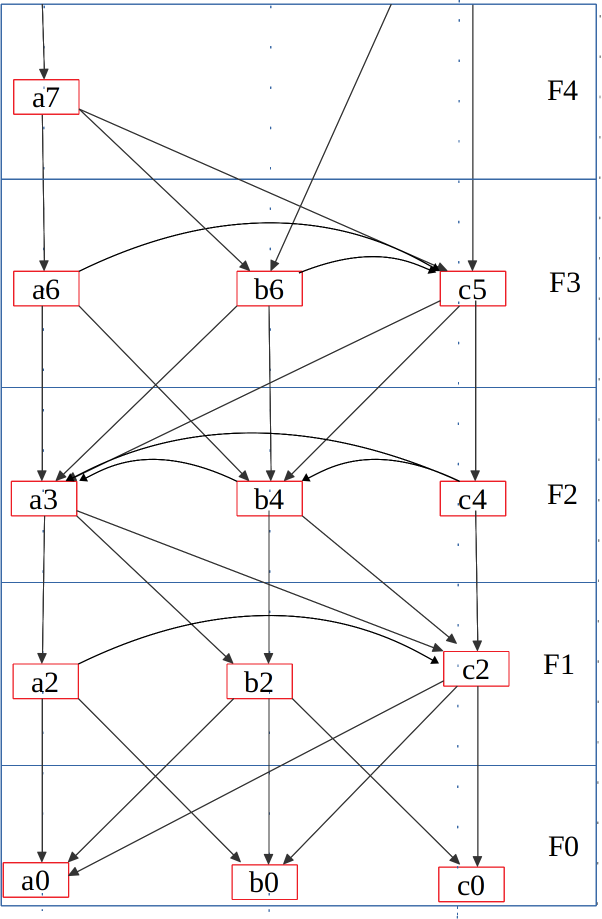}}
	\caption{An example of hierachical graph and root graph}
	\label{fig:layerframe}	
\end{figure}

\subsubsection{Frame Assignment}\label{se:frameassignment}

We then present a determinstic approach to  frame assignment, which assigns each event block a unique frame number.
First, we show how to assign frames to root vertices via the so-called \emph{root-layering}. Second, we then assign non-root vertices to the determined frame numbers with respect to the topological ordering of the vertices.

\begin{defn}[Root-layering]
	A root layering assigns each root with a unique layer (frame) number.	
\end{defn}

Here, we apply root layering on the root graph $G_R=(V_R, E_R)$. Then, we use root layering information to assign frames to non-root vertices. A frame assignment assigns every vertex $v$ in H-OPERA DAG $H$ a frame number $\phi_F(v)$ such that 
\begin{itemize}
	\item $\phi_F(u) \geq \phi_F(v)$ for every directed edge $(u,v)$ in $H$;
	\item for each root $r$, $\phi_F(r)= \phi_R(r)$;
	\item for every pair of $u$ and $v$: if $\phi(u) \geq \phi(v)$ then $\phi_F(u) \geq \phi_F(v)$;  if $\phi(u) = \phi(v)$ then $\phi_F(u) = \phi_F(v)$.
\end{itemize}
 Figure~\ref{fig:rootgraph-3-layer} depicts an example of root layering for the root graph in Figure~\ref{fig:rootgraph-3}.

\subsubsection{Consensus}

For an OPERA DAG $G$, let $G[v]$ denote the subgraph of $G$ that contains nodes and edges reachable from $v$.

\dfnn{Consistent chains}{For two chains $G_1$ and $G_2$, they are consistent if for any event $v$ contained in both chains, $G_1[v] = G_2[v]$. Denoted by $G_1 \sim G_2$.}

\dfnn{Global OPERA DAG}{A global consistent chain $G^C$ is a chain such that $G^C \sim G_i$ for all $G_i$.}

The layering of consistent OPERA DAGs is consistent itself.

\dfnn{Consistent layering}{For any two consistent OPERA DAGs $G_1$ and $G_2$,  layering results $\phi^{G_1}$ and $\phi^{G_2}$ are consistent if $\phi^{G_1}(v) = \phi^{G_2}(v)$, for any  vertex $v$ common to both chains. Denoted by $\phi^{G_1} \sim \phi^{G_2}$.}

\begin{thm}	
	For two consistent OPERA DAGs $G_1$ and $G_2$, the resulting H-OPERA DAGs using layering $\phi_{LPL}$ are consistent.
\end{thm} 

The theorem states that for any event $v$ contained in both OPERA DAGs, $\phi_{LPL}^{G_1}(v) = \phi_{LPL}^{G_2}(v)$. Since $G_1 \sim G_2$, we have $G_1[v]= G_2[v]$. Thus, the height of $v$ is the same in both $G_1$ and $G_2$. Thus, the assigned layer using $\phi_{LPL}$ is the same for $v$ in both chains.

\begin{prop}[Consistent root graphs]
	Two root graphs $G_R$ and $G'_R$ from two consistent H-OPERA DAGs are consistent.
\end{prop}

We can have similar proofs for the consistency of roots, clothos, and atropos as in Section~\ref{pf:root}~\ref{pf:clotho}, ~\ref{pf:atropos}.

\subsubsection{Main procedure}
Algorithm~\ref{al:main} shows the main function serving as the entry point to launch \stakedag\ protocol. The main function consists of two main loops, which are they run asynchronously in parallel. The first loop attempts to request for new nodes from other nodes, to create new event blocks and then communicate about them with all other nodes. 
The second loop will accept any incoming sync request from other nodes. The node will retrieve updates from other nodes and will then send responses that consist of its known events.

\begin{algorithm}[H]
	\caption{\stakedag\ Main Function}\label{al:main}
	\begin{tabular}{ll}
		\begin{minipage}{0.5\textwidth}
			\begin{algorithmic}[1]
				\Function{Main Function}{}
				\BState \emph{loop}:
				\State Let $\{n_i\}$ $\leftarrow$ $k$-PeerSelectionAlgo()
				\State Sync request to each node in  $\{n_i\}$
				\State (SyncPeer) all known events to each node in $\{n_i\}$
				\State Create new event block: newBlock($n$, $\{n_i\}$)
				\State (SyncOther) Broadcast out the message
				\State Update DAG $G$
				\State Call ComputeConsensus($G$)
				\BState \emph{loop}:
				\State Accepts sync request from a node
				\State Sync all known events by \stakedag\ protocol
				\EndFunction
			\end{algorithmic} 
		\end{minipage}
		\begin{minipage}{0.5\textwidth}
			\begin{algorithmic}[1]
				\Function{ComputeConsensus}{DAG}
				\State Apply layering to obtain H-OPERA DAG*
				\State Update flagtable
				\State Compute validation score*
				\State Root selection
				\State Compute root graph*
				\State Compute global state
				\State Clotho selection
				\State Atropos selection
				\State Order vertices and assign consensus time
				\State Compute Main chain
				\EndFunction	
			\end{algorithmic}
		\end{minipage}
	\end{tabular}
\end{algorithm}

\subsection{\stair\ Model}\label{se:model}

We present a model to allow low-staked accounts to join as a validating node, as described in \stair~\cite{stairdag}. In this model, we distinguish nodes into users (low-staked), validators (high-staked) and monitors (zero-staked). Users have less than $U$ tokens, but can join the network as a validator node. We give a summary to highlights the idea of \stair's model.

{\bf Event Block Creation}: Unlike StakeDag, \stair\ requires one of the other-ref blocks is from a creator of an opposite type. That is, if $n_i$ is a User, then one of the other-ref parents must be from a Validator block. If $n_i$ is a Validator, then one of the other-ref parents must be a User block.

{\bf x-DAG: Cross-type DAG}: \stair\ protocol maintains the DAG structure x-DAG, which is based on the concept of S-OPERA DAG in our StakeDag protocol~\cite{stakedag}. 
x-Dag chain is a weighted directed acyclic graph $G$=($V$,$E$), where $V$ is the set of event blocks, $E$ is the set of edges. Each vertex (or event block) is associated with a \emph{validation score}, which is the total weights of the roots reachable from it. 
When a block becomes a root, it is assigned a \emph{weight}, which is the same with the validating power $p_i$ of the creator node.

\begin{figure}[ht]
	\centering
	\includegraphics[width=0.3\linewidth]{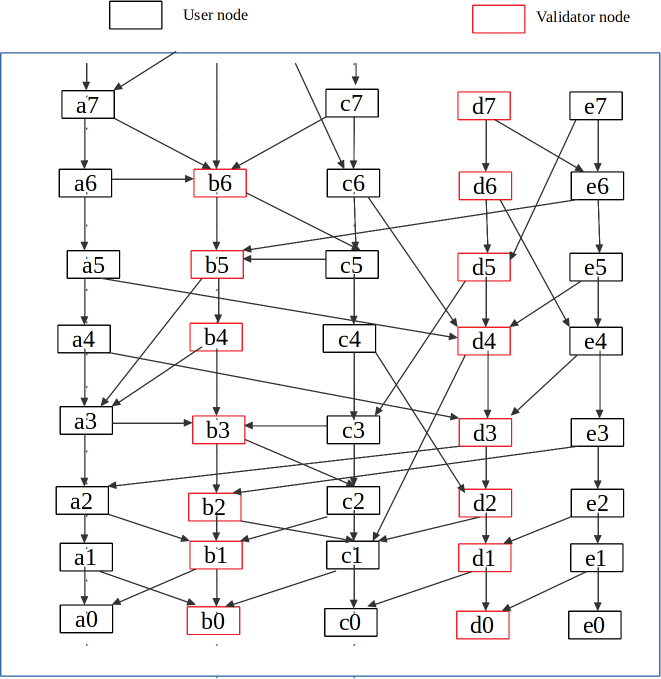}	
	\caption{Examples of x-DAG with users and validators}
	\label{fig:stakedag-ex}
\end{figure}

Figure~\ref{fig:stakedag-ex} gives an example of x-DAG, which is a weighted acyclic directed graph stored as the local view of each node. There are five nodes in the example, three of which are normal users with validating power of 1, and the rest are validators. Each event block has $k$=2 references. User blocks are colored Black and Validator blocks are in Red. Each Red block has one Red parent (same creator) and one Black parent (created by a User). Each Black block has one Black parent (same creator) and one Red parent (created by a Validator). As depicted, the example x-DAG is a RedBlack k-ary (directed acyclic) graph.

{\bf Stake-aware Peer Synchronization}: 
In \stair\ protocol, we present a approach for stake-aware peer synchronization to improve the fairness, safety and liveness of the protocol. \stair\ protocol enforces User nodes (low stake) to connect with Validators (high stake) in order to gain more validating power for User's new own blocks. Validators are also required to connect User nodes so as to validate the User's new blocks. Algorithm~\ref{al:main} shows the main function serving as the entry point to launch \stair\ protocol. 
\begin{algorithm}[H]
	\caption{\stair\ Main Function}\label{al:main}	
	\begin{algorithmic}[1]
		\State $\mathcal{U}$ $\leftarrow$ set of users ($t_i < U$).
		\State $\mathcal{V}$ $\leftarrow$ set of validators ($t_i \geq U$).
		\Function{Main Function}{$n_m$}
		\BState \emph{loop}:
		\State Let $\{n_i\}$ $\leftarrow$ $k$-PeerSelectionAlgo($n_m$, $\mathcal{U}$, $\mathcal{V}$)
		\State Sync request to each node in  $\{n_i\}$
		\State (SyncPeer) all known events to each node in $\{n_i\}$
		\State Create new event block: newBlock($n_m$, $\{n_i\}$)
		\State (SyncAll) Broadcast out the message		
		\State Update x-DAG $G$
		\State Call ComputeConsensus($G$)
		\BState \emph{loop}:
		\State Accepts sync request from a node
		\State Sync all known events by \stair\ protocol
		\EndFunction
	\end{algorithmic} 
\end{algorithm}

{\bf Saga points}: To measure the performance of participants, we propose a new notion, called Saga points (Sp). The Saga points or simply 'points' of a node $n_i$, is denoted by $\alpha_i$. As an example, Sp point is defined as the number of own events created by $n_i$ that has a parent (other-ref) as a finalized event (Atropos).

In \stair, we propose to use Saga points to scale the final value of validating power, say $w'_i = \alpha_i . w_i$.
Saga point is perfectly a good evidence for the validation achievement made by a node. Saga point is hard to get and nodes have to make a long time commitment to accumulate the points.

\clearpage
\section{Reference}\label{se:ref}

\renewcommand\refname{\vskip -1cm}
\bibliographystyle{abbrv}
\bibliography{Lachesis}

\begin{thebibliography}{10}

\bibitem{randomized03}
J.~Aspnes.
\newblock Randomized protocols for asynchronous consensus.
\newblock {\em Distributed Computing}, 16(2-3):165--175, 2003.

\bibitem{hashgraph16}
L.~Baird.
\newblock Hashgraph consensus: fair, fast, byzantine fault tolerance.
\newblock Technical report, 2016.

\bibitem{bang2008digraphs}
J.~Bang-Jensen and G.~Z. Gutin.
\newblock {\em Digraphs: theory, algorithms and applications}.
\newblock Springer Science \& Business Media, 2008.

\bibitem{bentov2016}
I.~Bentov, A.~Gabizon, and A.~Mizrahi.
\newblock Cryptocurrencies without proof of work.
\newblock In {\em International Conference on Financial Cryptography and Data
  Security}, pages 142--157. Springer, 2016.

\bibitem{bitshares}
{BitShares}.
\newblock Delegated proof-of-stake consensus.

\bibitem{buterin2013ethereum}
V.~Buterin et~al.
\newblock Ethereum white paper.
\newblock {\em GitHub repository}, 1:22--23, 2013.

\bibitem{Castro99}
M.~Castro and B.~Liskov.
\newblock Practical byzantine fault tolerance.
\newblock In {\em Proceedings of the Third Symposium on Operating Systems
  Design and Implementation}, OSDI '99, pages 173--186, Berkeley, CA, USA,
  1999. USENIX Association.

\bibitem{algorand16}
J.~Chen and S.~Micali.
\newblock {ALGORAND:} the efficient and democratic ledger.
\newblock {\em CoRR}, abs/1607.01341, 2016.

\bibitem{fantom18}
S.-M. Choi, J.~Park, Q.~Nguyen, and A.~Cronje.
\newblock Fantom: A scalable framework for asynchronous distributed systems.
\newblock {\em arXiv preprint arXiv:1810.10360}, 2018.

\bibitem{lachesis01}
S.-M. Choi, J.~Park, Q.~Nguyen, K.~Jang, H.~Cheob, Y.-S. Han, and B.-I. Ahn.
\newblock Opera: Reasoning about continuous common knowledge in asynchronous
  distributed systems, 2018.

\bibitem{byteball16}
A.~Churyumov.
\newblock Byteball: A decentralized system for storage and transfer of value,
  2016.

\bibitem{Coffman1972}
E.~G. Coffman, Jr. and R.~L. Graham.
\newblock Optimal scheduling for two-processor systems.
\newblock {\em Acta Inf.}, 1(3):200--213, Sept. 1972.

\bibitem{Ark}
T.~A. Crew.
\newblock Ark whitepaper.

\bibitem{Blockmania18}
G.~Danezis and D.~Hrycyszyn.
\newblock Blockmania: from block dags to consensus, 2018.

\bibitem{EOS}
EOS.
\newblock {EOS.IO} technical white paper.

\bibitem{ProofofAuth}
Ethcore.
\newblock Parity: Next generation ethereum browser.

\bibitem{algorand17}
Y.~Gilad, R.~Hemo, S.~Micali, G.~Vlachos, and N.~Zeldovich.
\newblock Algorand: Scaling byzantine agreements for cryptocurrencies.
\newblock In {\em Proceedings of the 26th Symposium on Operating Systems
  Principles}, pages 51--68. ACM, 2017.

\bibitem{king2012ppcoin}
S.~King and S.~Nadal.
\newblock Ppcoin: Peer-to-peer crypto-currency with proof-of-stake.
\newblock 19, 2012.

\bibitem{zyzzyva07}
R.~Kotla, L.~Alvisi, M.~Dahlin, A.~Clement, and E.~Wong.
\newblock Zyzzyva: speculative byzantine fault tolerance.
\newblock {\em ACM SIGOPS Operating Systems Review}, 41(6):45--58, 2007.

\bibitem{kwon2014tendermint}
J.~Kwon.
\newblock Tendermint: Consensus without mining.
\newblock {\em https://tendermint.com/static/docs/tendermint.pdf}, 2014.

\bibitem{lamport1978time}
L.~Lamport.
\newblock Time, clocks, and the ordering of events in a distributed system.
\newblock {\em Communications of the ACM}, 21(7):558--565, 1978.

\bibitem{paxos01}
L.~Lamport et~al.
\newblock Paxos made simple.
\newblock {\em ACM Sigact News}, 32(4):18--25, 2001.

\bibitem{Lamport82}
L.~Lamport, R.~Shostak, and M.~Pease.
\newblock The byzantine generals problem.
\newblock {\em ACM Trans. Program. Lang. Syst.}, 4(3):382--401, July 1982.

\bibitem{dpos14}
D.~Larimer.
\newblock Delegated proof-of-stake (dpos), 2014.

\bibitem{raiblock17}
C.~LeMahieu.
\newblock Raiblocks: A feeless distributed cryptocurrency network, 2017.

\bibitem{dagcoin15}
S.~D. Lerner.
\newblock Dagcoin, 2015.

\bibitem{conflux18}
C.~Li, P.~Li, W.~Xu, F.~Long, and A.~C.-c. Yao.
\newblock Scaling nakamoto consensus to thousands of transactions per second.
\newblock {\em arXiv preprint arXiv:1805.03870}, 2018.

\bibitem{Mehlhorn84a}
K.~Mehlhorn.
\newblock {\em Data Structures and Algorithms 2: Graph Algorithms and
  NP-Completeness}, volume~2 of {\em {EATCS} Monographs on Theoretical Computer
  Science}.
\newblock Springer, 1984.

\bibitem{honey16}
A.~Miller, Y.~Xia, K.~Croman, E.~Shi, and D.~Song.
\newblock The honey badger of bft protocols.
\newblock In {\em Proceedings of the 2016 ACM SIGSAC Conference on Computer and
  Communications Security}, pages 31--42. ACM, 2016.

\bibitem{bitcoin08}
S.~Nakamoto.
\newblock Bitcoin: A peer-to-peer electronic cash system, 2008.

\bibitem{onlay19}
Q.~Nguyen and A.~Cronje.
\newblock {ONLAY: Online Layering for scalable asynchronous BFT system}.
\newblock {\em arXiv preprint arXiv:1905.04867}, 2019.

\bibitem{stairdag}
Q.~Nguyen, A.~Cronje, M.~Kong, A.~Kampa, and G.~Samman.
\newblock Stairdag: Cross-dag validation for scalable {BFT} consensus.
\newblock {\em CoRR}, abs/1908.11810, 2019.

\bibitem{stakedag}
Q.~Nguyen, A.~Cronje, M.~Kong, A.~Kampa, and G.~Samman.
\newblock {StakeDag: Stake-based Consensus For Scalable Trustless Systems}.
\newblock {\em CoRR}, abs/1907.03655, 2019.

\bibitem{cck92}
P.~Panangaden and K.~Taylor.
\newblock Concurrent common knowledge: defining agreement for asynchronous
  systems.
\newblock {\em Distributed Computing}, 6(2):73--93, 1992.

\bibitem{panarello2018survey}
A.~Panarello, N.~Tapas, G.~Merlino, F.~Longo, and A.~Puliafito.
\newblock Blockchain and iot integration: A systematic survey.
\newblock {\em Sensors}, 18(8):2575, 2018.

\bibitem{pass2017fruitchains}
R.~Pass and E.~Shi.
\newblock Fruitchains: A fair blockchain.
\newblock In {\em Proceedings of the ACM Symposium on Principles of Distributed
  Computing}, pages 315--324. ACM, 2017.

\bibitem{PARSEC18}
F.~H. Q. M. S.~S. Pierre~Chevalier, Bartomiej Kamin~ski.
\newblock Protocol for asynchronous, reliable, secure and efficient consensus
  (parsec), 2018.

\bibitem{tangle17}
S.~Popov.
\newblock The tangle, 2017.

\bibitem{Sedgewick2011}
R.~Sedgewick and K.~Wayne.
\newblock {\em Algorithms}.
\newblock Addison-Wesley Professional, 4th edition, 2011.

\bibitem{sheikh2018proof}
H.~H.~S. Sheikh, R.~M.~R. Azmathullah, and F.~R.~R. HAQUE.
\newblock Proof-of-work vs proof-of-stake: A comparative analysis and an
  approach to blockchain consensus mechanism.
\newblock 2018.

\bibitem{sompolinsky2016spectre}
Y.~Sompolinsky, Y.~Lewenberg, and A.~Zohar.
\newblock Spectre: A fast and scalable cryptocurrency protocol.
\newblock {\em IACR Cryptology ePrint Archive}, 2016:1159, 2016.

\bibitem{PHANTOM08}
Y.~Sompolinsky and A.~Zohar.
\newblock Phantom, ghostdag: Two scalable blockdag protocols, 2008.

\bibitem{Steem}
Steem.
\newblock Steem white paper.

\bibitem{ppcoin12}
S.~N. Sunny~King.
\newblock Ppcoin: Peer-to-peer crypto-currency with proof-of-stake, 2012.

\bibitem{vasin2014blackcoin}
P.~Vasin.
\newblock Blackcoin’s proof-of-stake protocol v2.
\newblock {\em URL:
  https://blackcoin.co/blackcoin-pos-protocol-v2-whitepaper.pdf}, 71, 2014.

\bibitem{buterin2018}
B.~Vitalik and G.~Virgil.
\newblock Casper the friendly finality gadget.
\newblock {\em CoRR}, abs/1710.09437, 2017.

\end{thebibliography}

\end{document}